\documentclass[12pt]{article}

\usepackage{mathpazo} 

\usepackage{graphicx} 
\usepackage{setspace}
\usepackage[rgb]{xcolor}
\usepackage{verbatim}
\usepackage{tabularx}
\usepackage{subcaption}
\usepackage{amsgen,amsmath,amstext,amsbsy,amsopn,amsthm,tikz,bbm,bm}
\usepackage{fancyhdr}
\usepackage[colorlinks=true, urlcolor=blue,  linkcolor=blue, citecolor=blue]{hyperref}
\usepackage[colorinlistoftodos]{todonotes}
\usepackage{rotating}
\usepackage[shortlabels]{enumitem}
\usepackage{newpxtext} 
\usepackage{newpxmath} 
\usepackage{natbib}

\usepackage[capitalize, noabbrev]{cleveref}
\crefname{equation}{}{}

\usepackage{geometry}
\theoremstyle{plain}
\newtheorem{theorem}{Theorem}
\crefname{theorem}{Theorem}{Theorems}

\crefname{conjecture}{Conjecture}{Conjectures}
\newtheorem{proposition}{Proposition} 
\crefname{proposition}{Proposition}{Propositions}
\newtheorem{corollary}{Corollary} 
\crefname{corollary}{Corollary}{Corollaries}
\newtheorem{lemma}{Lemma} 
\crefname{lemma}{Lemma}{Lemmas}
\newtheorem{example}{Example}
\crefname{example}{Example}{Examples}
\newtheorem{definition}{Definition}
\crefname{definition}{Definition}{Definitions}

\theoremstyle{remark}
\newtheorem*{remark}{Remark}

\crefname{appendix}{Appendix}{Appendices}
\crefname{section}{Section}{Sections}
\crefname{figure}{Figure}{Figures}
\crefname{equation}{Equation}{Equations}

\newcommand{\E}{\mathbb{E}}

\title{Cued to Queue: \linebreak Information in Waiting-Line Auctions\thanks{
We thank Itai Ashlagi, Dirk Bergemann, Modibo Camara, Matthew Gentzkow, Yannai Gonczarowski, Oliver Hart, Ravi Jagadeesan, Zi Yang Kang, Scott Duke Kominers, Shengwu Li, Annie Liang, Eric Maskin, Paul Milgrom, Michael Ostrovsky, Andrzej Skrzypacz, Tomasz Strzalecki, Takuo Sugaya, Bob Wilson, Frank Yang, Weijie Zhong. We also thank seminar and conference participants at Brown, Harvard, Stanford, Stony Brook, and WZB Berlin for helpful discussions and suggestions. Tang thanks the NSF Graduate Research Fellowship for financial support. Any errors are our own.
}}

\author{
Jack Hirsch\thanks{Department of Economics, Harvard University, Cambridge, MA 02138. \href{mailto:jhirsch1@g.harvard.edu}{jhirsch1@g.harvard.edu}.}
\and
Eric Tang\thanks{Graduate School of Business, Stanford University, Stanford, CA 94305. \href{mailto:ertang@stanford.edu}{ertang@stanford.edu}.}
\vspace{0.5cm}
}

\date{\today}

\begin{document}

\maketitle
\begin{abstract}
We study the effect of providing information to agents who queue before a scarce good is distributed at a fixed time. Many information policies reveal \emph{sudden bad news}, when agents learn the queue is longer than previously believed. Sudden bad news causes assortative inefficiency by prompting multiple agents to simultaneously join the queue. If the value distribution has an increasing (decreasing) hazard rate, information policies that release sudden bad news increase (decrease) total surplus, relative to releasing no information. If agents incur entry costs and the hazard rate is decreasing, the optimal policy reveals only when the queue is full. (\emph{JEL} \text{D44, D61, D82}) 
\end{abstract}

\vfill

\pagebreak

\section{Introduction}

Queues are often used to distribute a scarce good at a fixed time.  
Agents choose when to join the queue, trading off their probability of obtaining the good with the time they waste in line. The result is a ``waiting-line auction,'' in which agents pay with time instead of money.  
For example, when a food pantry or soup kitchen has limited rations to distribute, individuals may queue well in advance of the opening time to receive a meal. Prior to the release of the iPhone X, some prospective customers queued for days outside of Apple stores to secure the device.\footnote{\href{https://www.businessinsider.com/iphone-x-people-queue-at-apple-stores-for-days-2017-11}{Business Insider (2017)}: ``In London and elsewhere in the world, queues of die-hard Apple fans are forming, prepared to camp out for a night or more to buy the eagerly anticipated tenth anniversary Apple smartphone'' (\cite{new_pric17}). \textit{See also} \href{https://www.businessinsider.com/apple-says-to-line-up-early-at-stores-if-you-want-an-iphone-x-2017-10}{Business Insider (2017)}, quoting an Apple press release: ``iPhone X will be available... in Apple Stores beginning Friday, November 3 at 8:00 a.m. local time. Stores will have iPhone X available for walk-in customers, \textit{who are encouraged to arrive early}'' [emphasis added] (\cite{news_lesw17}).} Homeless shelters allocating beds, theatres selling popular concert tickets, and gas stations selling rationed fuel face similar queues. Our paper studies how access to real-time information about the number of people in line changes the welfare properties of such queues.  

As a running example, suppose that a food pantry plans to distribute 50 parcels of food at noon, and more than 50 individuals desire a parcel. These agents face a trade-off: the earlier they arrive, the higher their probability of obtaining the good, but the more time they waste in line. Agents balance these two forces to determine when they will join the queue. However, information may alter their decisions. If an agent learns that there are already 45 people in line an hour before the pantry opens, she may hurry to join the queue before the line reaches 50 people. Conversely, if she discovers that there are only five people in line, she may wait to join until it is nearly noon, since more supply remains. Indeed, agents may rely on such information to guide their decisions, using social media or word of mouth to learn about the current queue length.

The distributor of the good may also provide information. For example, the food pantry could release a social media post updating individuals
when the line has reached a certain length, or they could post the length of the line every
hour. More generally, technology such as social media and live cameras may enable distributors to provide real-time updates about the length of the queue. The choice of whether to release information, and what information to release, changes individuals' queueing behavior.

In our model, the distributor of the good does not change the allocation scheme: the good is always allocated to the agents who arrive earliest, with ties broken randomly. In many settings, the distributor may not have the flexibility to alter the allocation scheme due to fairness concerns or concerns about the presence of speculators. In such cases, revealing real-time information (e.g., making an online announcement) may also be logistically simpler than altering the allocation procedure (e.g., running a lottery among all agents present). The distributor can influence allocation probabilities by committing to an information policy that controls the flow of public information that agents receive about the queue length. Returning to our example, if the food pantry commits to announce the length of the line every hour, individuals may modify their queueing behavior in anticipation of, and in response to, the incoming updates. 
We study how information affects the equilibrium allocation of the good, how long individuals wait, and total welfare.

\cref{thm:FOSD-assortative-efficiency} proves that under a condition we call \emph{sudden bad news}, providing real-time information results in assortative inefficiency: agents with the highest values are not guaranteed to receive the good. In equilibrium, agents form common beliefs about the number of individuals who have already joined the queue. When information leads to a jump in beliefs about the length of the queue, we say that agents received \emph{sudden bad news}. Many natural real-world information policies have a positive probability of releasing sudden bad news. These include information policies that (i) announce the length of the queue at a fixed time, (ii) announce when the queue reaches a certain length, or (iii) announce every time an agent joins the queue.

The intuition behind our main result is simple: if agents suddenly learn that the queue is longer than expected, a ``rush'' of agents may simultaneously attempt to join. If the number of agents rushing the queue exceeds the number of items remaining, then the objects are randomly allocated among those agents, breaking assortative efficiency. In contrast, if the information policy commits to never release information, the equilibrium is assortatively efficient. Thus, by controlling the flow of information about the length of the queue, the distributor can affect both assortative efficiency and the total waiting times of agents.

\cref{thm:FOSD-assortative-efficiency} shows that a wide range of information policies induce assortative inefficiency. While \cite{grrawa14} study a similar question in an auction setting, they analyze only the ``information policy'' in which agents learn each time someone purchases an item (for us, each time someone enters the queue). We instead employ a novel proof technique to analyze a wide range of information policies. Our proof uses the Payoff Equivalence theorem to abstract away from potentially complex equilibria induced by information policies, which allows us to simplify our analysis to sudden bad news. Our result also sharpens a general intuition for why real-time information causes inefficiency: ``bad news'' can drive many agents to suddenly rush the queue. 

Furthermore, we derive welfare implications that differ from those of the auction setting in \cite{grrawa14}. In an auction, assortative efficiency maximizes total surplus. 
A queue, on the other hand, is a money-burning mechanism; the ``payment'' is time spent waiting in line, which only reduces total surplus. As a result, total surplus is equivalent to consumer surplus, and assortative efficiency may not maximize consumer surplus. \cref{thm:FOSD-assortative-efficiency} suggests that a distributor should consider the impact of releasing such information, and we provide welfare results to guide her decision. We use \cref{thm:FOSD-assortative-efficiency} and existing work on money-burning mechanisms (e.g., \cite{haro08}, \cite{cond12}) to partially characterize when information policies increase or decrease total surplus.

Concretely, \cref{thm:information-efficiency} states that if the distribution of values for the good has an increasing hazard rate (as in the uniform distribution), then any information policy that can release sudden bad news \emph{increases} total surplus, relative to releasing no information. If the value distribution has a decreasing hazard rate (as in a Pareto distribution), then any such information policy \emph{decreases} total surplus. Intuitively, sudden bad news can cause assortative inefficiency, but it can also reduce agents' wait times. The hazard rate determines which effect dominates the welfare analysis. When the value distribution has a thick tail, it is more important to screen for high-value agents; when it has a thin tail, it is more important to reduce wait times.

\cref{Extensions} models a queue in which agents face an entry cost when attempting to join. For example, agents may incur a cost in time or money to travel to the distribution site. In such settings, an information designer may have further reason to release information: information can prevent agents from wastefully paying the entry cost without obtaining the good. \cref{cor:entry-cost-total-surplus} shows that when the hazard rate of the value distribution is decreasing, an information designer maximizes total welfare by announcing only when the queue is full. Indeed, when the hazard rate is strictly decreasing, this information policy yields a strictly higher surplus than either releasing no information or the fully informative policy of announcing every time an agent joins the queue. Thus to maximize surplus, an information designer must selectively release an ``intermediate'' amount of information. 

The paper is organized as follows. 
\cref{sec:related-literature} describes related literature. \cref{Model} introduces our model and formally demonstrates the equivalence between queues and strictly descending multi-unit auctions. \cref{sec:results} shows that sudden bad news causes assortative inefficiency. \cref{Optimality} partially characterizes how information affects total surplus. \cref{Extensions} considers an extension with entry costs. \cref{Conclusion} concludes. All omitted proofs are in the appendix. 

\subsection{Related Literature}\label{sec:related-literature}

\subsection*{Multi-Unit Descending Auctions}
Our model is similar to a multi-unit descending auction. \cite{hosh82} is the first paper to demonstrate that some queues, which they call ``waiting-line auctions,'' resemble Dutch auctions. However, they do not consider how information affects queueing behavior. The rush of agents that enters the queue following sudden bad news in our model is similar in spirit to \cite{bukl94}. They show that if prices continuously fall, information about a sale can induce agents to purchase immediately. However, they maintain assortative efficiency by allowing prices to increase to clear the market, whereas we model the “price” or entry time as strictly descending, reflecting the real-life constraints of queues. 

The work most similar to ours is \cite{grrawa14}, which studies a strictly-descending multi-unit Dutch auction. The price at which each object is sold is publicly announced, and therefore may induce a rush of agents. \cite{grrawa14} prove that any equilibrium is assortatively inefficient, and derive a symmetric equilibrium in a simple setting characterized by a set of initial value problems. 

While the auction model of \cite{grrawa14} is similar to our queue model, they only consider the full-disclosure information policy, which reveals each time an agent joins the queue. Our results apply to a much broader set of information policies, such as messages at a fixed time or probabilistic messages. The proof technique in \cite{grrawa14} does not extend to these information policies. Our work develops a new technique that uses the Payoff Equivalence theorem, which also provides a novel intuition for the underlying cause of assortative inefficiency.\footnote{In particular, we show that the phenomenon of ``sudden bad news'' leading to assortative inefficiency is a much more general occurrence than in the full-disclosure information policy.} Furthermore, because we analyze a queueing setting instead of an auction setting, information has different implications for welfare. We provide distinct results on how information affects total surplus, as well as associated policy recommendations for information designers.

\subsection*{Queueing Models}

Many queueing models consider stochastic agent arrivals and a limited capacity to process arriving agents. The seminal early work is \cite{naor69}, and \cite{haha03} provide a survey. More recently, \cite{lesh22} considers a setting with stochastic arrivals of agents who have heterogeneous preferences between two over-demanded items with waitlists, and demonstrates that information policies can increase total surplus. Related work on dynamic assignment includes \cite{baleya20}, \cite{aniyma23} and \cite{muth24}. \cite{chte24} consider a dynamic queue in which homogeneous agents arrive stochastically and decide whether to join the queue. In their model, the queue designer controls agents' entry and exit from the queue, their queue priority, and the information they receive.

We model a fundamentally different type of queue. In our setting, all agents are present at the beginning of the game, all items are allocated at the same time, and each agent makes a single decision about when to join the queue. Our model describes a line that forms before a distributor (such as a food pantry or store) allocates goods at a fixed opening time. The dynamic queueing literature is more apt to model a setting such as a public housing waitlist, in which the waitlist's length fluctuates as new housing becomes available and new agents arrive. In contrast, a queue forms in our setting not because agents arrive faster than items arrive, but instead because agents choose to join the queue before the distribution time. As a result, in our model, agents choose \emph{when} to join the queue; in the work cited above, agents generally choose \emph{whether} to join a queue. Consequently, we employ analytic tools from auction theory rather than the stationary distributions used in the dynamic queueing literature. To our knowledge, our paper is the first to study the effects of information on queues that form due to a fixed distribution time and a fixed population of agents.

\subsection*{Money-Burning Mechanisms}

A queue is an example of a ``money-burning'' or ``ordeal'' mechanism, in which agents are screened by a costly ordeal that does not transfer surplus (as opposed to screening via money). In such settings, there is a trade-off between wasteful costly screening and assortative efficiency. 
\cite{dwor23} compares screening with ordeal mechanisms to lump-sum transfers and \cite{tats03} compares costly screening to lotteries, but neither work analyzes the role of information in queues. 
\cite{haro08} and \cite{cond12} construct consumer surplus-maximizing ordeal mechanisms for allocation problems. \cref{thm:information-efficiency} combines \cref{thm:FOSD-assortative-efficiency} and the results in these papers to partially characterize how information policies change total surplus in our queueing model.

\section{Model}\label{Model}

\subsection{Setup}

Consider $n$ risk-neutral agents indexed by $i\in\{1, \ldots, n\}$. Each agent has unit demand for one of $k < n$ identical copies of an item. 
Agents have private values drawn independently and identically from distribution $F$, which we assume is atomless and has full support on $[0, \bar{v}]$.\footnote{All results extend to distributions with unbounded support, so long as the expected value of the distribution is finite. If the expectation of the distribution is infinite, all results hold except for the strict welfare comparisons in \cref{thm:information-efficiency} and \cref{cor:entry-cost-total-surplus}. \cref{app:unbounded-support} provides further details.} We fix model primitives $n$, $k$, and $F$ unless stated otherwise, and we assume that they are all common knowledge.  
Throughout the paper, $Y_k^{(n)}\sim F^{(n)}_k$ denotes the $k^{\text{th}}$ order statistic of $n$ draws from $F$.

Instead of paying to receive a copy of the good, agents queue in line. We assume that the utility of agents is linear in time. We further assume that each agent has the same disutility for each unit of time spent waiting in line, normalized to $1$.\footnote{If agents have heterogeneous waiting costs, one can normalize values to value per time waited. However, the implications for total welfare will then differ.} 

Time $t\in\mathbb{R}_+$ continuously descends and represents the time remaining until the good is distributed, at time $t = 0$. If there are $q$ agents in the queue, we say there are $k - q$ \emph{items remaining}, although no agent receives the good until $t=0$.

At any time, an agent may attempt to join the queue. An agent who attempts to join the queue observes the queue's length, then decides whether to stay in the queue or leave. We assume that if an agent attempts to join the queue, she cannot leave and later rejoin. An agent therefore \emph{successfully} joins the queue if there are still items remaining ($q < k$). If more agents attempt to join the queue than items remain, ties are broken uniformly at random. We also assume that attempting to join the queue is costless, so an agent who unsuccessfully attempts to join the queue obtains payoff $0$.\footnote{We allow for entry costs in \cref{Extensions}.} We discuss the rationale for these assumptions in \cref{subsec:game-procedure-discussion}.

Each agent $i$ therefore obtains utility
\begin{align*}
    u(v_i, t) &= 
    \begin{cases}
        v_i - t & \text{ if $i$ successfully joins at $t$} \\
        0 & \text{ otherwise.}
    \end{cases}
\end{align*}

The information designer (also called the announcer) commits to a commonly known information policy. A message released at time $t$ instantaneously becomes common knowledge to all agents at time $t$. Natural information policies include announcing every time an agent joins the queue or announcing the length of the queue at a particular time. We let $\mathcal{Q}(p)$ denote the queueing game with information policy $p$.

\subsection{Game Procedure} 

We model the queue as a descending clock game. As the clock descends, agents may attempt to join the queue and messages may arrive. Our model is complicated by the timing of interactions between messages and agent arrivals. We therefore formalize the model in a sequential format, with a finite number of ``bidding'' rounds. Throughout the paper, we interchangeably refer to the time at which an agent attempts to join the queue as his bid, often in the context of employing results from auction theory.

To summarize our formal model, in each round, the announcer submits the next time she would like to send a message (conditional on no agent joining the queue before that time), and each agent submits the next time he would like to join the queue (conditional on no message arriving before that time). The next message or queue entry is processed, with ties broken in favor of the announcer. Each time a message arrives, agents may respond by instantaneously attempting to join the queue. When an agent joins the queue, the announcer may instantaneously release a new message. We use she/her pronouns for the announcer and he/him pronouns for agents.

Let $T$ denote the current game clock (which is descending). The order of play is as follows:

\textbf{Bidding Stage}. Each agent $i$ who has not joined the queue updates his sealed bid $b_i\in[0, T)$ if there has been a new message since the previous bidding stage. The announcer draws her next message and message time $(m, t)$ with $t \in [0, T)$. The announcer's draw does not depend on the sealed bids.

\textbf{Resolution Stage}. Let $b:=\max_i \{b_i\}$ denote the highest bid. Update the clock to $T = \max\{b, t\}$. If $b > t$, then the highest bidder(s) join the queue. Otherwise, message $m$ is released. The following steps then repeatedly alternate at time $T$: 

\begin{itemize}
    \item If a message has arrived, bidders may (instantaneously) attempt to join the queue. If no bidder does so, the game returns to the bidding stage.
    \item If a bidder has attempted to join the queue, the announcer may (instantaneously) send a message. If she does not, the game returns to the bidding stage.
\end{itemize}

When $T=0$, items are allocated and the game ends.

\begin{remark}
    Once the queue is full, any agent not in the queue will obtain payoff $0$, regardless of his actions. This is because any agent not yet in the queue will not obtain the good and therefore would never choose to wait in line. For convenience, and without loss of generality, we describe agents as learning when the queue is full.

    In \cref{Extensions}, we model agents who pay an entry cost to join the queue.
    In that case, whether or not the information designer announces when the queue is full affects agents' strategies.
\end{remark}

\subsection{Discussion of Game Procedure}\label{subsec:game-procedure-discussion}

    Our baseline model assumes that attempting to join the queue is costless. In \cref{Extensions} we incorporate entry costs for joining the queue, such as time or money spent traveling to the queue. We also assume that when an agent attempts to join the queue, he immediately learns the queue length and decides whether or not to stay. In many cases, an agent joining a queue can observe how many others are currently in line to determine if the queue is full. 
    
    If an agent leaves, however, we assume he cannot rejoin at a later time. In many real-world settings, we find it unlikely that agents typically plan multiple trips to and from the queue.\footnote{For example, a prospective concert attendee may choose one time to travel to the concert venue, without entertaining the thought of traveling between the venue and their home multiple times.} There are also two practical reasons we forbid queue re-entry. First, our paper focuses on the effects of public information on queues, and re-entry allows agents to acquire private information. Second, agent behavior can become quite complex if agents can leave and rejoin the queue. For example, an agent who joins the queue and observes that there are few other agents currently in the queue may find it optimal to leave the queue and rejoin at a later time. For these reasons, we simplify agents' behavior by assuming they cannot rejoin the queue.

    Finally, ties between agents are broken uniformly at random. For example, if two items remain and three agents attempt to join the queue simultaneously, each has a two-thirds chance of success. This assumption is motivated by the reality of physical queues. In most real-world queues, if several agents arrive to a queue at the same time, their order in line is essentially random. Stated differently, agents who arrive simultaneously can no longer be screened by their willingness to wait. Our uniform tiebreaking assumption could also model a short, stochastic delay between when agents attempt to join the queue and when they actually arrive at the queue.

\subsection{Information Policies}

In the main text we describe information policies and our equilibrium concept with verbal descriptions. We provide mathematical formulations in the appendix.

The \emph{message history} captures the content of past messages and the time at which they were sent, and is visible to both the announcer and the agents. The \emph{queue history} captures the history of past arrivals to the queue, and is visible only to the announcer. We refer to the joint message history and queue history as the \emph{history}. 

\begin{definition}[Information Policy]    
    An \emph{information policy} specifies the time and content of publicly announced messages as a function of the history and current time.
\end{definition}

We further require that information policies can send only a finite number of messages (this is a technical assumption to ensure the game ends). The baseline information policy that we use for comparison is the ``trivial'' information policy that never sends a message.

\begin{example}[Trivial]\label{ptrivial}
    The announcer never sends a message. 
\end{example}
We also include two natural examples of substantive information policies. Note that information policies can be probabilistic, so that messages may or may not arrive.
\begin{example}[Fixed Time]\label{pfixedtime}
    The announcer reveals the current queue length at fixed time $\tau$, where $\tau$ is common knowledge at the beginning of the game. 
\end{example}

\begin{example}[Full Revelation]\label{peachentry}
    The announcer reveals the updated queue length each time an agent joins the queue. As a result, agents always know the length of the queue. 
\end{example}

\subsection{Equilibrium}

We study Perfect Bayesian Equilibria, defined by agents' strategies and beliefs. Each agent follows a strategy that maps his value, the current message history, and the current time to an entry time. In equilibrium, all agents derive beliefs about the the number of items remaining, determined by the information policy and model primitives, the message history, and the current time. 

Without loss of generality, we may limit our attention to agents' beliefs about the \textit{positive} quantities of items remaining, $\mu \in \Delta(\{1, \ldots, k\})$—recall that if the queue is full, an agent obtains payoff $0$ regardless of her action. Her best response is therefore determined only by $\mu$.
Since information policies are common knowledge and release public messages, all agents who have not yet joined the queue share a common belief.

After a message arrives at some time $t$, agents' beliefs may change. We therefore index agents' beliefs in a way that allows us to describe ``limiting beliefs'' before and after a message can arrive. We let $\mu_{t^+}$ denote agents' beliefs about the number of remaining items at time $t$, but \emph{before} the announcer has had the opportunity to send any message at time $t$.
We let $\mu_{t^-}$ denote agents' beliefs about the number of remaining items at time $t$ immediately \emph{after} the announcer has had the opportunity to send a message. When necessary, we continue in this fashion by using $\mu_{t^{--}}$, etc. to index further rounds of the nested game occurring at the same time-step.\footnote
{See \cref{sec:appendix-definitions} for a formal definition of $\mu_{t^+}$ and $\mu_{t^-}$. 
We use $t^{--}$ to index beliefs after the announcer has had two chances to send messages at time $t$. That is, $t^{--}$ indexes beliefs immediately after the second instance that the announcer has the opportunity to send a message in the Resolution Stage at time $t$, if this step is reached. We then use $t^{---}$ for the third instance that the announcer has the opportunity to send a message at time $t$, and so forth.
} 

Our main analysis studies the equilibrium properties of queueing games defined by various information policies. We focus on total surplus and assortative efficiency.

\begin{definition}[Assortatively Efficient]
    An equilibrium is \emph{assortatively efficient} if the $k$ agents with highest values are always allocated the good.
\end{definition}

Since a queue is a money-burning mechanism, time spent waiting in line is lost surplus. As a result, an assortatively efficient equilibrium does not necessarily maximize total surplus: the screening costs necessary to ensure assortative efficiency may outweigh the surplus gains of allocating the items to the agents with the highest values. In \cref{Optimality}, we analyze how information policies affect total surplus.

\section{Results}\label{sec:results}

\subsection{Benchmark: No Information}

Consider the trivial information policy from \cref{ptrivial}. Since there are no messages, agents never learn new information, so the game is strategically equivalent to each agent writing down the time at which he will enter the queue. This is the same as choosing a bid in a multi-unit sealed-bid discriminatory price auction.\footnote{I.e., a multi-unit auction in which each winner pays his bid.}
As in a discriminatory price auction, agents trade off their probability of winning with their surplus conditional on winning. The only difference between the queue model and the auction model is the form of payment: money or time.

\cite{reny99} proves that a multi-unit sealed-bid discriminatory price auction has a symmetric, assortatively efficient equilibrium characterized by the bidding function\footnote{As in our setting, \cite{reny99} considers agents with unit demand and in which values are private and drawn independently and identically.} 

\begin{align}\label{discrimbid}
    \beta(v) &= \mathbb{E}[Y_k^{(n-1)}| Y_k^{(n-1)} < v].
\end{align}

In any symmetric equilibrium of this auction, agents cannot pool.\footnote{If agents pool, agents in the pooling region could profitably deviate by increasing their bid by some $\varepsilon >0$.} Hence any symmetric equilibrium is assortatively efficient, and by the Payoff Equivalence theorem, payments must be defined by (\ref{discrimbid}). This is therefore the unique symmetric equilibrium.

\begin{corollary}\label{cor:trivial-info-policy}
    If information policy $p$ is trivial, then the unique symmetric equilibrium of $\mathcal{Q}(p)$ is assortatively efficient. 
\end{corollary}

In our benchmark, releasing no information results in assortative efficiency, as the agents with highest values join the queue earliest. \cref{thm:FOSD-assortative-efficiency} shows that a wide range of information policies instead result in an assortatively inefficient allocation.

\subsection{Main Results}

To analyze how information policies affect the ultimate allocation, we define a notion of ``sudden bad news'' and show that many natural information policies release sudden bad news with positive probability. \cref{thm:FOSD-assortative-efficiency} then shows that if an equilibrium has a positive probability of releasing sudden bad news, it is assortatively inefficient.

\begin{definition}\label{def:sudden-bad-news}
    We say \emph{sudden bad news} is revealed at time $t$ if (i) $\mu_{t^-}$ is strictly FOSD dominated by $\mu_{t^+}$ and (ii) $\mu_{t^+}$ is weakly FOSD dominated by $\mu_{t'}$ for all $t' > t$. 
\end{definition}

Condition (i) holds when agents learn that there are strictly fewer objects remaining than they had previously believed. Intuitively, this leads agents to bid more aggressively, but the current time $t$ imposes a cap on how aggressively agents can bid. Condition (ii) ensures that $\mu_{t^-}$ is the ``most pessimistic belief'' agents have had in the game thus far.\footnote{Note that sudden bad news can also occur in the absence of a message. For example, suppose the designer commits to send one message at time $\tau$ if there are fewer than five agents in the queue. If no message arrives at time $\tau$, agents learn there are at least five agents in the queue, and their beliefs update discontinuously.}
Many intuitive information policies have a positive probability of releasing sudden bad news in any equilibrium:

\begin{itemize}
    \item \cref{pfixedtime} (fixed time). The message has a positive probability of revealing that one item remains, which FOSD dominates any previous belief. Even if the message occurs probabilistically or after a delay, there is still a positive probability of sudden bad news. 
    
    \item \cref{peachentry} (full revelation). Every message is sudden bad news. Agents always know the exact queue length, so learning that a new agent entered is sudden bad news. Delays or probabilistic messages still result in sudden bad news.\footnote{If the messages occur probabilistically or with a delay, there is still a positive probability of announcing that only one item remains, which is sudden bad news.}
    
    \item An information policy that reveals when exactly $k'$ items remain. Before the message, agents know that there must be strictly more than $k'$ items remaining, so the message reveals sudden bad news.
\end{itemize}

\begin{theorem}\label{thm:FOSD-assortative-efficiency}
    Fix information policy $p$. Any equilibrium of $\mathcal{Q}(p)$ with a positive probability of revealing sudden bad news is assortatively inefficient.
\end{theorem}

Intuitively, sudden bad news reveals that the good is scarcer than previously believed, driving agents to bid more aggressively (join the queue earlier). However, since the clock is already at $\tau$, high-value agents cannot bid more than $\tau$, and some are forced to pool by bidding $\tau$. 
To formalize the intuition, we consider the game immediately after sudden bad news is revealed, and show that conditional payments in the paused queueing game are not equal to conditional payments in an assortatively efficient mechanism. Then by the Payoff Equivalence theorem, the allocation of the queueing game cannot be assortatively efficient.

To prove \cref{thm:FOSD-assortative-efficiency}, we first develop a series of lemmas. Fix information policy $p$. By the revelation principle, any equilibrium of $\mathcal{Q}(p)$ corresponds to an incentive-compatible direct mechanism. We will use the Payoff Equivalence theorem to compare the equilibrium payments of $\mathcal{Q}(p)$ to the equilibrium payments of an assortatively efficient mechanism. We first formally demonstrate that the Payoff Equivalence theorem holds when there is supply and demand uncertainty. 

\begin{lemma}\label{RET:supply_demand_uncertain}[Payoff Equivalence Theorem]
    Fix agents' beliefs about the number of remaining items and agents. Then any two mechanisms with the same allocation rule have the same expected payments for all types.
\end{lemma}

We use \cref{RET:supply_demand_uncertain} to characterize a symmetric equilibrium bidding function for a discriminatory price auction when agents are uncertain about the number of items remaining and the number of other agents.\footnote{One consequence of \cref{RET:supply_demand_uncertain} is that if an equilibrium of queueing game $\mathcal{Q}(p)$ has the same allocation rule as an equilibrium of queueing game $\mathcal{Q}(p')$, then the two equilibria have the same expected payments for all types.} Once the game begins, agents form beliefs about the remaining supply and demand, pinned down by their beliefs about the number of objects remaining. Since our main proof technique considers intermediate stages of the game, it is useful to establish various properties when there is demand and supply uncertainty. In this case, the equilibrium bidding function is given by the weighted average of the bidding functions in \cref{discrimbid}. However, the weighted average is not just over the probability the state occurs, but also over the probability the agent wins in that state. 

\begin{proposition}\label{prop:dpa_bid}
    Consider a discriminatory auction with unit-demand and uncertainty about the number of items available and the number of bidders. There exists an assortatively efficient equilibrium in which agents bid the weighted average of the equilibrium bidding functions for discriminatory price auctions with known supply and demand, given by \cref{bid_func:sup_dem_uncertain}.
\end{proposition}

The proof and exact bidding function (\cref{bid_func:sup_dem_uncertain}) are in the appendix. We use \cref{prop:dpa_bid} to characterize transfers in the queueing game with more complex information policies.

\begin{definition}
    Suppose agents have beliefs $\mu$ about the number of items remaining and suppose agents have values drawn independently and identically from $F|_{[0, \hat{v}]}$. Let $t_{AE}(v | \mu, \hat{v})$ denote the expected payment of type $v$ conditional on winning in an assortatively efficient mechanism. Let $t_{Q}(v | \mu, \hat{v}, p)$ denote the same for an equilibrium of the queueing game with information policy $p$. 
\end{definition}

We suppress $p$ in the notation for $t_Q$ when clear.

\begin{lemma}\label{lem:t_AE-continuous}
    The expected transfer of the highest remaining type in an assortatively efficient mechanism, $v \mapsto t_{AE}(v | \mu, v)$, is continuous and strictly increasing.
\end{lemma}

\cref{lem:t_AE-continuous} says that if we fix beliefs $\mu$ and values are drawn independently and identically from some $F|_{[0, v]}$, then the expected transfer in an assortatively efficient mechanism for an individual with the highest possible type $v$ is continuous in $v$. Note that changing $v$ also changes the support over which values are drawn. The next lemma describes how expected payments in an assortatively efficient mechanism change as beliefs change.

\begin{lemma}\label{lem:t_AE-FOSD-dominance} 
    Suppose that belief $\mu$ strictly FOSD dominates $\mu'$. Then $t_{AE}(v | \mu', v) > t_{AE}(v |\mu, v)$ for all $v$.
\end{lemma}

Because $\mu$ FOSD dominates $\mu'$, there are fewer items available under $\mu'$. As a result, an agent with the highest type must pay more conditional on winning (and always wins) in equilibrium. In the equilibrium of a discriminatory price auction, such an agent would bid more aggressively. 

We use $t_{AE}$ to pin down the time at which each type must join the queue, given their beliefs as the game progresses.

\begin{lemma}\label{lem:transfer-determines-entry-time}
    Suppose $\mathcal{Q}(p)$ has an assortatively efficient equilibrium. In this equilibrium, type $v$ attempts to enter the queue at some $\tilde{t}' \in \{t^+, t^-, t^{--}, ...\}$ if and only if $t_{AE}(v | \mu_{\tilde{t}}, v) = t$ for some $\tilde{t} \in \{t^+, t^-, t^{--}, ...\}$.
\end{lemma}

In an assortatively efficient equilibrium, an agent's payment conditional on winning is determined by his beliefs, his value, and the highest type remaining in the game. This in turn determines the time at which he attempts to enter the queue, because that time is also his payment conditional on winning.

We use \cref{RET:supply_demand_uncertain}, \cref{lem:t_AE-continuous}, \cref{lem:t_AE-FOSD-dominance}, and \cref{lem:transfer-determines-entry-time} to show that if information policy $p$ has a positive probability of sudden bad news, then transfers in $\mathcal{Q}(p)$ cannot be equal to transfers in an assortatively efficient mechanism. To prove \cref{thm:FOSD-assortative-efficiency}, we assume that $\mathcal{Q}(p)$ results in an assortatively efficient allocation but also reveals sudden bad news. We then show that after sudden bad news, the highest type remaining in the game would need to pay more than the current clock time. 

\begin{proof}[Proof of \cref{thm:FOSD-assortative-efficiency}]

Fix an equilibrium of $\mathcal{Q}(p)$ with a positive probability of revealing sudden bad news. For the sake of contradiction, assume that the equilibrium is assortatively efficient. Suppose sudden bad news is revealed at time $t$.

Consider two paused versions of $\mathcal{Q}(p)$: (i) the game immediately before sudden bad news is revealed, when agents have beliefs $\mu_{t^+}$ and (ii) the game immediately after, when agents have beliefs $\mu_{t^-}$.\footnote
{
More precisely, for (i), pause the game at time $t$ and ask the agents their beliefs before they know whether or not the announcer sent a message at time $t$. For (ii), tell agents the announcer's message (if there was one) and again elicit their beliefs. 
} 
Let $\hat{v}$ denote the highest type that could still be in the game before the message (determined by agents' equilibrium bidding strategies), so that agents have types drawn from $F|_{[0, \hat{v}]}$. Since the equilibrium is assortatively efficient, only one type can join in any resolution stage of the nested game. Thus after the message the remaining agents have types in $[0, \hat{v})$ or $[0, \hat{v}]$.

Because the allocation is assortatively efficient by assumption, the queueing game continued from any point must result in an assortatively efficient allocation. By Lemma \ref{RET:supply_demand_uncertain}, each type's expected transfer conditional on winning in the queueing game must be equal to their expected transfer conditional on winning in an assortatively efficient mechanism, given their respective beliefs. In notation, $t_Q(v | \mu, \hat{v}) = t_{AE}(v | \mu, \hat{v})$ for any $v$, $\mu$ and $\hat{v}$. 

We claim $t_{AE}(\hat{v} | \mu_{t^+}, \hat{v}) = t$. Assume not. 
Since the clock is already at time $t$, type $\hat{v}$ cannot pay more than $t$, so $t_{AE}(\hat{v} | \mu_{t^+}, \hat{v}) = t_{Q}(\hat{v}|\mu_{t^+},\hat{v}) \leq t$. Hence $t_{AE}(\hat{v} | \mu_{t^+}, \hat{v}) < t$. By \cref{lem:t_AE-continuous}, we may choose $v' > \hat{v}$ sufficiently close to $\hat{v}$ such that $t_{AE}(v' | \mu_{t^+}, v') < t$.\footnote{Note that the definition of sudden bad news implies $\hat{v} < \overline{v}$ so such a type exists.} We now show $v'$ cannot have already entered the queue, which contradicts the maximality of $\hat{v}$. By the definition of sudden bad news and \cref{lem:t_AE-FOSD-dominance}, for all $t' > t$ we know $t_{AE}(v' | \mu_{t'}, v') \leq t_{AE}(v' | \mu_{t^+}, v') < t < t'$.\footnote{Intuitively, our theorem condition states that these are the most pessimistic beliefs agents have had thus far in the game, hence all previous expected payments (conditional on winning) were lower than the current expected payment (conditional on winning).} Then by \cref{lem:transfer-determines-entry-time}, type $v'$ cannot have entered the queue before time $t$ and is still in the game. This contracts the maximality of $\hat{v}$. We conclude that $t_{AE}(\hat{v} | \mu_{t^+}, \hat{v}) = t$.

By \cref{lem:t_AE-FOSD-dominance}, $t_{AE}(\hat{v} | \mu_{t^-}, \hat{v}) > t_{AE}(\hat{v} | \mu_{t^+}, \hat{v}) = t$. By the continuity of $t_{AE}$, there exists some type $v < \hat{v}$ with $t_{AE}(v | \mu_{t^-}, \hat{v}) > t$. That means type $v$ is still in the game but satisfies
\begin{equation*}
    t_Q(v | \mu_{t^-}, \hat{v}) = t_{AE}(v | \mu_{t^-}, \hat{v}) > t.
\end{equation*}

A type-$v$ agent has positive probability of winning the item, and if she does, she must pay $ t_Q(v | \mu_{t^-}, \hat{v}) > t$ in expectation. Since the game clock has already reached $t$, her payment cannot exceed $t$, a contradiction.
\end{proof}

We have presented many examples of information policies that are assortatively inefficient, but it is also possible to have (non-trivial) information policies that maintain assortative efficiency. As suggested by \cref{thm:FOSD-assortative-efficiency}, the driving force behind assortative inefficiency is the instantaneous release of bad news, which causes a mass of agents to rush to join. An information policy can release bad news continuously and never cause a rush. 
\cref{ex:continuous-bad-news-info-policy} considers an information policy that $(1)$ provides bad news continuously and $(2)$ can only make an announcement that lowers agents' beliefs about the length of the queue.

\begin{example}[Continuous Bad News]\label{ex:continuous-bad-news-info-policy}
  Let $n=3$ and $k=2$. Suppose agents have values distributed independently from $U([0,1])$. Fix $\tau =2/9$. Consider information policy $p$ that (i) privately draws a time $\tau' \sim U([0,\tau])$ and (ii) at time $\tau'$ sends message $m$ if the queue was empty at time $\tau$, and otherwise sends no message.
\end{example}

Consider the game with information policy defined by \cref{ex:continuous-bad-news-info-policy}. As time passes and a message does not arrive, agents continuously update their beliefs about the queue length, with increasingly pessimistic views about the probability that the queue was empty at time $\tau$. If a message arrives, it signifies good news and relaxes entry functions.

\begin{proposition}\label{prop:cont-bad-news}
    Let $p$ be the information policy from \cref{ex:continuous-bad-news-info-policy}. Then $\mathcal{Q}(p)$ has an assortatively efficient equilibrium.
\end{proposition}

We choose $\tau = 2/9$ because the analysis is simpler for a fixed $\tau$.\footnote{We believe any $\tau\in(0, 1/3)$ will suffice, because $1/3$ is the earliest time an agent might join the queue.} \cref{prop:cont-bad-news} shows that substantive information policies can still yield assortative efficiency. 
Because \cref{ex:continuous-bad-news-info-policy} generates an assortatively efficient allocation, the total welfare is also the same as providing no information.
Our paper focuses on the effects of information policies on total welfare and assortative efficiency. As a result, using only the desiderata we consider, there is no distinction between different information policies that result in assortative efficiency. We therefore limit our analysis to a single example. 

\section{Welfare}\label{Optimality}

We next analyze how information affects total welfare, which can help predict the welfare impacts of information arriving from public sources such as social media and the internet. More importantly, our results may help an information designer increase total surplus with a well-designed information policy. Suppose the distributor of a good wishes to allocate a good using a queue instead of (for example) a lottery. This may be due to logistical concerns (a lottery may be more difficult to organize), concerns about allocation to very-high-need individuals (such as in the distribution of food aid), or concerns about the presence of speculators. In such settings, when an organization wishes to maintain a queue but is willing to release real-time information about that queue, our results can guide what information should be released.

\cref{thm:FOSD-assortative-efficiency} allows us to characterize some cases in which providing information increases or decreases total surplus. Intuitively, when information breaks assortative efficiency, it reduces the value generated by the allocation, but it also reduces total waiting times. The hazard rate of the value distribution determines which effect dominates the welfare analysis. If the distribution has a ``thick tail,'' maintaining assortative efficiency is more important; if the distribution has a ``thin tail,'' reducing waiting times is more important. Results similar to \cite{haro08} allow us to characterize when assortative efficiency maximizes or minimizes surplus. 

\begin{lemma}\label{lem:hazard-rate-efficiency}
    Consider the class of incentive compatible and individually rational money-burning mechanisms that allocate all items. If the hazard rate of the value distribution is weakly increasing (decreasing), then any mechanism with an assortatively inefficient allocation has a weakly higher (lower) total surplus than an assortatively efficient mechanism. If the hazard rate is strictly monotone, the inequality is strict.
\end{lemma}

To our knowledge, no prior work has connected this result to information design for waiting-line auctions. \cref{lem:hazard-rate-efficiency}, together with \cref{thm:FOSD-assortative-efficiency}, yields \cref{thm:information-efficiency}.

\begin{theorem}\label{thm:information-efficiency}
    Let $\tilde{p}$ be the trivial information policy and let $p$ be any information policy. Suppose $\frac{1-F}{f}$ is weakly increasing (decreasing). Any equilibrium of $\mathcal{Q}(p)$ with positive probability of revealing sudden bad news yields weakly lower (higher) total surplus than the symmetric equilibrium of $\mathcal{Q}(\tilde{p})$. If the hazard rate is strictly monotone, the inequality is strict.
\end{theorem}

\begin{proof}
    By \cref{cor:trivial-info-policy}, the symmetric equilibrium of the queueing game with a trivial information policy is assortatively efficient. Fix nontrivial information policy $p$ and an equilibrium of $\mathcal{Q}(p)$ which has a positive probability of releasing sudden bad news. By \cref{thm:FOSD-assortative-efficiency}, this equilibrium of $\mathcal{Q}(p)$ is assortatively inefficient. 
    
    By the revelation principle, any equilibrium of the queueing game corresponds to a Bayesian incentive-compatible and individually rational direct-revelation mechanism which yields the same allocation and payoffs for each agent as the equilibrium of the queueing game. In particular, the equilibrium of the queueing game with a trivial information policy corresponds to an incentive compatible and individually rational mechanism whose allocation is assortatively efficient. The equilibrium of $\mathcal{Q}(p)$ corresponds to an incentive compatible and individually rational mechanism whose allocation is not assortatively efficient. Both mechanisms always allocate all items. If the hazard rate of the value distribution is increasing, then by \cref{lem:hazard-rate-efficiency}, a mechanism with an assortatively efficient allocation minimizes total surplus among mechanisms that allocate all items. Hence this equilibrium of $\mathcal{Q}(p)$ yields weakly higher total surplus than the equilibrium of the queueing game with trivial information policy. The results for a decreasing hazard rate and strict monotonicity follow identically from Lemma \ref{lem:hazard-rate-efficiency}.
\end{proof}

\cref{thm:information-efficiency} provides a policy recommendation for an information designer who knows the value distribution of her customers. If the hazard rate is increasing, she can increase consumer surplus by providing information, and if the hazard rate is decreasing, she maximizes consumer surplus by providing no information.\footnote{Consider the case of a decreasing hazard rate. If the equilibrium of a non-trivial information policy never releases sudden bad news, then it either yields lower expected surplus (if the allocation is assortatively inefficient) or it yields the same expected surplus as the symmetric equilibrium of the trivial information policy (if the allocation is assortatively efficient).} Our results also suggest that an information designer may find it worthwhile to learn the distribution of values of her customers to guide her decisions. 

As discussed in subsection \ref{subsec:game-procedure-discussion}, \cref{thm:information-efficiency} depends on several modeling assumptions about the queueing environment, including the assumption that agents can costlessly join the queue. In \cref{Extensions}, we develop a model that uses the same framework, but incorporates entry costs. We derive different implications for how an information designer should optimize total surplus. In particular, maximizing total surplus may require releasing an \textit{intermediate} level of information: disclosing when the queue is full, but not revealing sudden bad news.

\section{Entry Costs}\label{Extensions}

In many real-world situations, agents face an upfront cost to join a queue. When a food bank distributes aid or a store releases an over-demanded product, agents physically travel to join the queue, which can be costly in money or time. In the absence of information, agents risk incurring the entry cost without securing the good, which may reduce consumer surplus. Furthermore, some low-value agents may choose not to queue at all. An information designer may share information about the queue's length to avoid wasted entry costs or unallocated items, but as we showed in \cref{Optimality}, information can have subtle effects on welfare.

In this section, we show that if agents face entry costs, and the value distribution has a strictly decreasing hazard rate, then the information designer maximizes surplus by revealing only some information. In particular, \cref{cor:entry-cost-total-surplus} states that one optimal information policy is to reveal when the queue is full, and nothing else. The resulting equilibrium has strictly higher welfare than revealing no information, and also has strictly higher welfare than any policy that reveals sudden bad news (such as full revelation, revealing the queue length at a fixed time, or announcing when the queue reaches a certain length).

Suppose that all agents face a homogeneous entry cost $c \in (0, \overline{v})$ to join the queue. If an agent attempts to join the queue, they first pay cost $c$, then the game proceeds as in the main model. We let $\mathcal{Q}_c(p)$ denote the entry-cost queueing game with entry cost $c$ and information policy $p$.\footnote
{For tractability, we assume that the cost is paid instantly. If the cost is paid in time, there may be a delay between when an agent receives information and when they arrive at the queue. If the cost is paid in money, there is no delay, although consequences for total surplus may differ: what we describe as total surplus is now consumer surplus.} 
Without entry costs, announcing when the queue is full does not affect payoffs or the equilibrium outcome: if the queue is full, agents obtain zero utility regardless of when or whether they attempt to join the queue. With entry costs, if the queue is full, agents lose utility if they attempt to join. As a result, agents change their behavior if the announcer discloses when the queue is full.

Modeling entry costs also changes the assortatively efficient allocation. In the entry cost model, an equilibrium is assortatively efficient if the $k$ agents with highest values above $c$, or all such agents if there are fewer than $k$, receive the good. When there are entry costs, releasing no information is no longer assortatively efficient: some agents who could have profitably entered the queue choose not to, due to uncertainty about whether they will obtain the item.

\begin{proposition}\label{prop:entry-cost-no-information}
    Suppose agents face an entry cost $c \in (0,\bar{v})$ and a trivial information policy. Then in equilibrium, the good is allocated to the (up to) $k$ agents with highest values above some $v^R > c$. In particular, the allocation is not assortatively efficient.
\end{proposition}

To achieve assortative efficiency, we consider an information policy that discloses when the queue is full (and nothing else). By announcing when the queue is full, agents with values close to $c$ still have an incentive to join before they receive a message. 
\begin{proposition}\label{prop:entry-cost-ae}
    Suppose information policy $p$ discloses when the queue becomes full and sends no other messages. Then the unique symmetric equilibrium of $\mathcal{Q}_c(p)$ is assortatively efficient.
\end{proposition}

If the announcer reveals ``too much'' information, in the sense formalized in \cref{def:sudden-bad-news}, we again lose assortative efficiency due to a mass of agents potentially rushing the queue. Thus to maintain assortative efficiency, the announcer must release some information (such as disclosing when the queue is full) without disclosing sudden bad news. As described in \cref{sec:results}, a wide variety of natural information policies can therefore break assortative efficiency.

\begin{proposition}\label{prop:entry-cost-sudden-bad-news}
    Consider the entry-cost queue game with an information policy $p$ that always reveals when the queue is full, and may also send other messages. If in equilibrium $\mathcal{Q}_c(p)$ has a positive probability of revealing sudden bad news, then the equilibrium is assortatively inefficient.
\end{proposition}

If the announcer always discloses when the queue is full, we can reinterpret the game as a standard queueing game in which there are no entry costs and agents' values are reduced by $c$. Sudden bad news is still defined relative to beliefs over the \emph{positive} number of items remaining. \cref{thm:FOSD-assortative-efficiency} then implies that there is no assortatively efficient equilibrium. To connect assortative efficiency and welfare in the entry cost setting, we establish a lemma similar to the results in \cite{haro08}.

\begin{lemma}\label{lem:hazard-rate-efficiency-all-mechs}
    Suppose agents' values are distributed with support $[\underline{v}, \overline{v}]$ for $\underline{v} \leq 0$ and $\overline{v} \geq 0$. Consider the class of incentive compatible and individually rational money-burning mechanisms. If the hazard rate of the value distribution is weakly (strictly) decreasing, then an assortatively efficient mechanism yields a weakly (strictly) higher total surplus than any mechanism in this class that is not assortatively efficient.
\end{lemma}

The proof is similar to that of Theorem 2.9 of \cite{haro08}. The only differences are that (i) we allow for agents with negative values and (ii) we maximize surplus over all possible mechanisms, not only over mechanisms that always allocate the good.\footnote{Indeed, when $\underline{v} < 0$, individually rational mechanisms do not always allocate the good. Negative values model agents for whom the entry cost exceeds their value.}

Combining \cref{prop:entry-cost-no-information}, \cref{prop:entry-cost-ae}, and \cref{prop:entry-cost-sudden-bad-news} with \cref{lem:hazard-rate-efficiency-all-mechs} yields a condition under which the information designer prefers to release an intermediate amount of information.

\begin{corollary}\label{cor:entry-cost-total-surplus}
    Suppose the hazard rate of the value distribution is weakly decreasing. Let $p$ be an information policy that discloses when the queue is full and makes no other announcements. The symmetric equilibrium of $Q_c(p)$ yields weakly higher total surplus than any equilibrium of $Q_c(p')$ for any information policy $p'$. If the hazard rate is strictly decreasing and $p'$ is the trivial information policy, or the equilibrium of $Q_c(p')$ releases sudden bad news with positive probability, this inequality is strict.
\end{corollary}
Returning to our leading example of the food pantry, suppose that the hazard rate of the value distribution is indeed strictly decreasing, and agents face an entry cost to join the queue. Suppose that the food pantry wishes to maximize total surplus by altering the information it releases about the queue length. The food pantry can do so by announcing only when the queue is full; furthermore, this policy should be \emph{strictly} preferred to a wide range of other natural policies. It is strictly suboptimal for the food pantry to reveal no information. Yet it is also strictly suboptimal to fully reveal the length of the queue, to disclose the queue length at a fixed time, or to announce when the queue reaches a fixed length (other than $k$). These results show that information has a nuanced effect on total surplus, and may require organizations to carefully consider what information they release. Our results can guide organizations in making their decision.

\section{Conclusion}\label{Conclusion}

We study how access to information affects welfare when queues form in advance of the distribution of a good. Under a condition we call ``sudden bad news,'' releasing information leads to an assortatively inefficient allocation. We partially characterize when providing real-time information increases or decreases total surplus, and we show that the effect on total surplus depends on the hazard rate of the value distribution. When agents face an entry cost to join the queue, an information designer may maximize surplus by revealing when the queue is full and not revealing further information. 

One direction for future work is to determine the surplus-maximizing information policy for a given distribution of values. If the hazard rate is increasing, \cref{thm:information-efficiency} shows that releasing sudden bad news increases surplus relative to releasing no information, but even in simple examples it is difficult to determine the optimal policy. 

There are many ways to extend our model. When there is uncertainty about the initial supply of items, agents have multi-dimensional beliefs and our ``sudden bad news'' criterion no longer necessarily induces a rush. There are also settings with multiple queues: for example, homeless shelters in the same city may each have a queue, or multiple stores may distribute the same scarce good. One could also model a setting in which those earlier in line receive a higher-quality good or face a shorter processing time. Finally, future work could expand the set of permissible information policies and actions by modeling costly information acquisition, private messages to a subset of agents, or agents who can leave and re-enter the queue.

\newpage

\begin{appendix}

\section{Definitions}\label{sec:appendix-definitions}

The \emph{message history} $h^M = \big\{(m_i, t_i)\big\}_{i=1}^{j}$ captures the log of past messages, where $(m_i, t_i)$ signifies that the $i^{\text{th}}$ message $m_i$ was announced at time $t_i$. We denote the space of all possible message histories by $\mathcal{H}^M$. The message history is visible to both the announcer and the players. 

The \emph{queue history} $h^Q$ captures the history of past arrivals to the queue, and is visible only to the announcer. We denote a generic queue history as $h^Q = \{ t_1, ..., t_m \} \in \mathcal{H}^Q$, where $t_i$ denotes the time that the $i^{\text{th}}$ agent in the queue joined the queue.

Whereas the message history is public knowledge, the queue history is not. We therefore combine $h^M$ and $h^Q$ into one history $h \in\mathcal{H}$ only visible to the announcer that records the order of messages and queue entrances. 

\begin{definition}[Information Policy]    
    An information policy specifies the time and content of messages as a function of the history and current time. Let $\mathcal{M}$ denote an arbitrary message space. An \emph{information policy} is a function 
    \begin{align*}
        p: \mathbb{R}_+ \times \mathcal{H} \to \Delta(\mathcal{M} \times \mathbb{R}_+),
    \end{align*} 
    where $p(t, h)$ is the distribution of possible messages sent when the current time is $t$ and the history is $h$.\footnote{We choose to write the information policy as mapping onto a release time to be consistent with our sequential game procedure.} The announcer cannot send messages earlier than the current time: any tuple $(m', t')$ in the support of $p(t, h)$ has $t' \leq t$.\footnote{The message is only announced if it conforms with the time restrictions given in the game procedure: namely, $t' < t$ in the main stage of the game, or $t' \leq t$ in the nested game.} We also restrict information policies to send a finite number of messages over the course of the game. Most explicit information policies we consider are deterministic, in which case we write $p$ as a function mapping to the space $\mathcal{M} \times \mathbb{R}_+$.
\end{definition}

\begin{definition}[Strategy]
    A strategy $s_i$ for player $i$ is a map from the current message history, time, and player value to their entry time. Formally, 
    \begin{align*}
        s_i : \mathcal{H}^M\times \mathbb{R}_+ \times \mathbb{R}_+ &\rightarrow \mathbb{R}_+\\
        (h^M, t, v) &\mapsto b
    \end{align*}
    Each strategy implicitly assumes a fixed information policy $p$, so that $s_i(h^M, t, v | p) = b$.
\end{definition}

\begin{definition}[Equilibrium]
    We consider Perfect Bayesian Equilibrium (PBE). If strategy profile $S^*$ forms a PBE of $\mathcal{Q}(p)$, then
    \begin{enumerate}[label = (\roman*)]
        \item For each player $i$, strategy $s_i^*$ is a best response to $s_{-i}^*$, the common belief $\mu$, and current message history $h^M$.
        \item The on-path common belief $\mu\in\Delta(\{1, \ldots, n\})$ over the positive number of items remaining is correctly derived via Bayes' Rule.\footnote{See the main text for a discussion of why beliefs are only over the positive number of items remaining.} In particular, let $X_t$ be a random variable denoting the number of items remaining at time $t$. Let $\mu_t(j)$ denote the probability that $j>0$ items remain, conditional on $X_t > 0$. Then $\mu_t(j)$ is given by
        
    \end{enumerate}
        \begin{align}\label{eqn:history-bayes-update}
            \mu_t(j) := Pr(X_t = j \big| h^M, t,X_t>0) = \frac{Pr(h^M, t, X_t = j)}{Pr(h^M,t, X_t>0)},
        \end{align}
    
    When $Pr(h^M,t) = 0$, we replace the probability mass expressions with their respective probability densities.
\end{definition}

\begin{definition}[Limiting Beliefs ($\mu_{t^+}$)]
    Fix a time $t$. Consider the last Bidding Stage that occurred strictly before time $t$ in which bidders had the opportunity to update their bids. Let $\nu \in \Delta(\{1,\dots,k\})$ denote an agent's belief about the number of items currently remaining at this stage. Let $\psi \in \Delta(\{1, \dots,n\})$ denote her belief about the number of bidders who will submit bids strictly greater than $t$ in this Bidding Stage. (Note that $\psi$ is a property of the equilibrium bidding strategies.)

    Then define $\mu_{t+} \in \Delta(\{1, \dots, k\})$ by:

    $$
    \mu_{t^+}(j) := \frac{\sum_{j' = j}^{k} \nu(j') \psi(j'-j)}{\sum_{h=1}^{k}\sum_{h' = h}^{k} \nu(h') \psi(h'-h)}.
    $$

    That is, the belief that there are $j$ items remaining at $t^+$ sums over the probability that there were $j$ items remaining during the previous bidding stage and no one else joined the queue before time $t$, plus the probability that there were $j+1$ items remaining during the previous bidding stage and exactly one additional agent joined the queue before time $t$, and so forth.

    Informally, $\mu_{t^+}$ is an agent's belief in this bidding stage about the number of items currently remaining, adjusted downward by their belief about the number of agents who will submit bids strictly greater than $t$, conditional on no further messages being announced before time $t$. The denominator normalizes the belief to condition on a positive number of items remaining.
\end{definition}

\begin{definition}[Limiting Beliefs ($\mu_{t^-}$)]
    If the announcer sent at least one message at time $t$, then $\mu_{t^-} \in \Delta(\{1, 2, \dots,k\})$ is simply bidders' beliefs about the number of items remaining (conditional on a positive number of items remaining) after receiving the first message at time $t$, but before any agent can join the queue at time $t$.

    If the announcer did not send a message at time $t$, then $\mu_{t^-}$ is bidders' beliefs about the number of items remaining (that is, $k$ minus the number of agents who arrived strictly before time $t$) given that no message arrived at time $t$.

    When necessary, we continue in this fashion to index further rounds of the nested game occurring at the same time-step. We use $t^{--}$ to index beliefs after the announcer has had two opportunities to send messages at time $t$. That is, $t^{--}$ indexes beliefs immediately after the second instance that the announcer has the opportunity to send a message in the Resolution Stage at time $t$, if this step is reached. We then use $t^{---}$ for the third instance that the announcer has the opportunity to send a message at time $t$, and so forth.
\end{definition}

\section{Unbounded Support}\label{app:unbounded-support}

Our results also extend to distributions $F$ with support on $[0, \infty)$. Let $Y\sim F$. Then our results hold exactly as written so long as $\E[Y] < \infty$. If $\E[Y] = \infty$, then the ``strict'' welfare results in \cref{thm:information-efficiency} and \cref{cor:entry-cost-total-surplus} do not hold, because the expected total surplus from these equilibria are infinite. Note that extending our results to unbounded $F$ is a meaningful extension: the only distributions with monotonically decreasing hazard rates have unbounded support.

All results in \cref{sec:results} are valid as written. The only change to the game procedure is that the clock begins at time $T = \infty$, allowing agents to submit arbitrarily high bids in the first round, but the game procedure still has a finite number of rounds. No further regularity conditions are needed, because all expectations used in the proofs of \cref{sec:results} are conditional expectations of bounded random variables, and are therefore finite. In addition, the applications of the Envelope Theorem are still valid with unbounded support.

In \cref{Optimality} and \cref{Extensions}, the 
finite expectation ensures that the strict inequalities hold. For general distributions with unbounded support, if $\E[Y] = \infty$, then the expected total surplus is infinite as well. For example, consider the Pareto distribution with $\alpha \leq 1$. \cref{lem:hazard-rate-efficiency} shows that randomly allocating all items weakly minimizes surplus. However, that policy would yield expected surplus $ \frac{k}{n} \times \E[Y] = \infty$. Therefore an assortatively efficient mechanism cannot yield strictly higher surplus. 

For any distribution with finite expectation, all mechanisms have finite expected surplus. Any mechanism has a smaller expected surplus than giving \emph{every} agent a free copy of the good, which would yield expected total surplus of $n \E[Y] < \infty$. Thus so long as $\E[Y]$ is finite, our strict welfare comparisons are meaningful. Our regularity condition therefore states that all results hold so long as no potential mechanism yields infinite surplus.

Some examples of distributions with strictly decreasing hazard rate and finite expectation include the Pareto distribution with scale parameter $\alpha > 1$ and the Weibull distribution with scale parameter $\lambda > 0$ and shape parameter $\beta \in (0, 1)$.

\section{Examples of Information Policies}

\begin{definition}[Trivial]
    Information policy $p$ is trivial if 
    \begin{align*}
        \text{supp} \left\{ p(t, h)\right\} \subset \Delta(\mathcal{M}\times\{0\}).
    \end{align*}
\end{definition}
This means that messages can only arrive when $t = 0$. In this case, each agent chooses a single bid $b$ to maximize her expected profit, with no chance to revise her bid in the future.

One function to obtain the information policy $p$ in \cref{pfixedtime} is given by
    \begin{align*}
        p(t, h) = \begin{cases}
            (|h^Q|, \tau) & \; \text{if} \; t > \tau \\
            (0, 0) & \; \text{otherwise.}
        \end{cases}
    \end{align*}
    To see why, suppose the clock has not yet reached $\tau$ and there are $|h^Q|$ agents in the queue. If no agent joins the queue before $\tau$, the information policy will announce that the queue has length $|h^Q|$ at time $\tau$. If an agent does join before $\tau$, the policy updates the message accordingly. Also note that while $h^Q$ is not an explicit input of $p$, it can be recovered from $h$. 

One function to obtain the information policy $p$ in \cref{peachentry} is to announce whenever someone enters the queue. Then
    \begin{align*}
        p(t, h) = 
        \begin{cases}
            (|h^Q|, t) & \text{ if } h^M_{-1} \neq |h^Q| \\
            (0, 0) & \text{ otherwise,} 
        \end{cases}
    \end{align*}
    where $h^M_{-1}$ denotes the most recent message sent by the announcer. 
    
    Each time an agent enters the queue, a new message is drawn from $p$, and it will trigger a message to announce the current length of the queue. Afterward, $p$ reverts to subsequent message time of $t=0$, which corresponds to sending no message. This default state remains until another agent joins the queue. The policy ensures all agents know the current length of the queue at all times. 

\begin{example}[Fixed Time and State]\label{ptimestate}
    Say $p$ announces at time $\tau$ if and only if the queue contains exactly $k$ agents. Then
    \begin{align*}
        p(t, h) = 
        \begin{cases}
            (|h^Q|, t) & \text{ if } |h^Q| = k \text{ and } t \geq \tau\\
            (0, 0) & \text{ otherwise.} 
        \end{cases}
    \end{align*}
\end{example}
The intuition here is that the absence of an message can also be informative, and may change how agents behave after the time at which a message could have come. 

\begin{proof}[Proof of \cref{RET:supply_demand_uncertain}]
    Consider a mechanism with uncertainty about the number of potential agents and number of items to be allocated. Let $\mathcal{N} = \{\underline{n}, \ldots,  \bar{n}\}$ denote the set of potential agents and let $\mathcal{K} = \{\underline{k}, \ldots,  \bar{k}\}$ represent the set of potential number of goods available.
    We allow for arbitrary correlation, so (with slight abuse of notation) the probability that there are $n$ agents and $k$ items is given by $(n, k) \sim H \in \Delta(\mathcal{N} \times \mathcal{K})$ with probability mass function $h(n,k)$.\footnote{In the middle of a queueing game, the number of items remaining uniquely determines the number of competitors remaining. We use this lemma on ``paused'' moments in the middle of the game, and thus need to allow for correlation.}
    Recall that agents have values drawn independently and identically from atomless $F$. 
    All participating agents share common beliefs about the number of competitors and items for sale.

    By assumption, both mechanisms have the same probability $\tilde{G}(v; n, k)$ of allocating an item to an agent of type $v$ when there are $n$ agents and $k$ items. The overall probability that an agent reporting type $z$ wins an item is therefore
\begin{align*}
    G(z)
    &:= \sum_{n = \underline{n}}^{\bar{n}} \sum_{k = \underline{k}}^{\bar{k}} h(n,k) \tilde{G}(z; n, k).
\end{align*}
The rest of the proof is a standard application of the Payoff Equivalence theorem. Note that type $v=0$ obtains zero utility and makes zero payment in both mechanisms. The expected payment of type $v$ is then given by
\begin{align*}
    p(v_i) &= p(0) + \int_0^{v_i} v G'(v) dv
\end{align*}

which depends only on the allocation rule $G$.
\end{proof}

\begin{proof}[Proof of \cref{prop:dpa_bid}]
    As in \cref{RET:supply_demand_uncertain}, let $\mathcal{N} = \{\underline{n}, \ldots,  \bar{n}\}$ denote the set of potential agents, let $\mathcal{K} = \{\underline{k}, \ldots,  \bar{k}\}$ represent the set of potential number of goods available. Abusing notation slightly, suppose $(n, k) \sim H \in \Delta(\mathcal{N} \times \mathcal{K})$ with probability mass function $h(n,k)$.
    Agents have values drawn independently and identically from atomless $F$, and let $Y^{(n)}_k \sim F^{(n)}_k$ represent the random variable that is the $k^{\text{th}}$ highest of $n$ draws from $F$.\footnote{If $k > n$, we take $Y^{(n)}_k \equiv 0$.} 
    All participating agents share common beliefs about the number of competitors and items for sale. 
    
     First consider a unit-demand uniform-price auction. For any $n$ and $k$, if an agent faces $n$ other bidders and there are $k$ items, it is an equilibrium to bid truthfully.\footnote{See e.g. \cite{kris09}, chapter 13.} Consequently, even when there is uncertainty over agents and items, it is still optimal to bid truthfully.
     An agent with value $v$ wins if $v > Y^{(n-1)}_k$, which occurs with probability $F^{(n-1)}_k(v)$. The overall probability she wins is therefore
\begin{align*}
    G(v)
    &:= \sum_{n = \underline{n}}^{\bar{n}} \sum_{k = \underline{k}}^{\bar{k}} h(n,k) F^{(n-1)}_k(v),
\end{align*}
    and the expected payment is 
    \begin{align}\label{paymentUPA}
        p^{UPA}(v) &= \sum_{n = \underline{n} }^{\bar{n}} \sum_{k = \underline{k}}^{\bar{k}} h(n,k)  F^{(n-1)}_k(v) \E[Y^{(n-1)}_k | Y^{(n-1)}_k < v].
    \end{align}
    Next consider a discriminatory price auction. Let $\beta$ be a symmetric, increasing equilibrium. For an agent with value $v$, the expected payment is
    \begin{align}\label{paymentDPA}
        p^{DPA}(v) &= G(v) \beta(v).
    \end{align}
    By the Payoff Equivalence theorem (\cref{RET:supply_demand_uncertain}), the payments in Equations $\eqref{paymentUPA}$ and $\eqref{paymentDPA}$ are equal. Hence
    \begin{align}\label{bid_func:sup_dem_uncertain}
        \notag G(v) \beta(v) &= \sum_{n = \underline{n}}^{\bar{n}} \sum_{k = \underline{k}}^{\bar{k}} h(n,k) F^{(n-1)}_k(v) \E[Y^{(n-1)}_k | Y^{(n-1)}_k < v]\\
        \beta(v) &= \sum_{n = \underline{n}}^{\bar{n}} \sum_{k = \underline{k}}^{\bar{k}} \frac{h(n,k)  F^{(n-1)}_k(v)}{G(v)} \times \E[Y^{(n-1)}_k | Y^{(n-1)}_k < v].
    \end{align}
    We conclude that the bidding function in a discriminatory price auction with supply and demand uncertainty is a weighted average of the bidding functions with known supply and demand, weighted by the probability of the state times the probability the agent wins in the state. 
\end{proof}

\begin{proof}[Proof of \cref{lem:t_AE-continuous}]
    Let $\mathcal{N} = \{\underline{n}, \ldots,  \bar{n}\}$ denote the set of potential agents and let $\mathcal{K} = \{\underline{k}, \ldots,  \bar{k}\}$ represent the set of potential number of goods available. 
    Agents' beliefs that $n$ agents and $k$ items remain are given by $(n, k) \sim H \in \Delta(\mathcal{N} \times \mathcal{K})$ with probability mass function $h(n,k)$. Any belief $\mu$ determines one such distribution $h$; in this distribution, $h(n,k)>0$ implies $n > k$ (because there are initially more agents than items).\footnote{To be precise, if there are initially $N$ agents and $K$ items, then the distribution is defined by $h(N - K + k, k) = \mu(k)$.}
    Agents have values drawn independently and identically from $\bar{F} = F|_{[0, v]}$, denoted by random variable $\bar{Y}$. 
    By \cref{RET:supply_demand_uncertain}, for any assortatively efficient mechanism a type-$v$ agent has expected transfer given by \cref{paymentUPA}:
    \begin{align*}
        t_{AE}(v|\mu, v) &= \sum_{n = \underline{n} }^{\bar{n}} \sum_{k = \underline{k}}^{\bar{k}} h(n,k)  \bar{F}^{(n-1)}_k(v) \E[\bar{Y}^{(n-1)}_k | \bar{Y}^{(n-1)}_k < v].
    \end{align*}
    Since $v$ is the highest type, they win with probability $1$, and the above expression simplifies to 
    \begin{align*}
        t_{AE}(v|\mu, v) &= \sum_{n = \underline{n} }^{\bar{n}} \sum_{k = \underline{k}}^{\bar{k}} h(n,k)  \E[\bar{Y}^{(n-1)}_k].
    \end{align*}
    It is thus sufficient to show $\E[\bar{Y}^{(n-1)}_k]$ is continuous and strictly increasing in $v$. 
    Since $\bar{F}$ is a truncated version of $F$, as the upper bound of the support increases, so does the expectation of each draw. For continuity, note that 
    \begin{align*}
        \E[\bar{Y}^{(n-1)}_k] &= \int_0^v x \frac{\partial}{\partial x}\left[\bar{F}^{(n-1)}_{k}(x)\right]dx\\
        &= v \bar{F}^{(n-1)}_{k}(v) - \int_0^v \bar{F}^{(n-1)}_{k}(x)dx \\
        &= v - \int_0^v \left[\sum_{j = 0}^{k-1} \binom{n-1}{j}\left(\frac{F(x)}{F(v)}\right)^{n-1-j}
        \left(1 - \left(\frac{F(x)}{F(v)}\right)\right)^j\right]dx.
    \end{align*}
    The second line integrates by parts, and the third line uses the fact that $\bar{F}^{(n-1)}_{k}(v) = 1$ and expands the term under the integral. 
    Since $F$ is atomless, it is continuous, hence the integrand is continuous, as is the whole expression.\footnote{As $v\rightarrow 0$, the support of $\bar{F}$ collapses to $0$, and the payment is also $0$.} 
\end{proof}

\begin{proof}[Proof of \cref{lem:t_AE-FOSD-dominance}]
     Recall that the game begins with $n$ agents and $k$ items available. If $\kappa$ items remain, there must be $n - k + \kappa$ remaining agents. Because $v$ is the highest remaining type, agents have values drawn independently and identically from $F|_{[0, v]}$. Consider a uniform-price auction: the goods are allocated to the $\kappa$ bidders with the highest bids at the price of the $(\kappa+1)^{\text{st}}$-highest bid. It is a weakly dominant strategy to bid truthfully, in which case the auction is assortatively efficient. An agent with value $v$ then has an expected payment conditional on winning of $\E[Y_{\kappa}^{(n - k + \kappa - 1)} | Y_{\kappa}^{(n - k + \kappa - 1)} < v] = \E[Y_{\kappa}^{(n - k + \kappa - 1)}]$, where $Y$ is a random variable drawn from $F|_{[0, v]}$. Averaging the expected payment conditional on winning across all states,
    \begin{align*}
        t_{AE}(v | \mu, v) &= \frac{\sum_{\kappa=1}^k \mu(\kappa) \cdot Pr(Y_{\kappa}^{(n - k + \kappa - 1)} < v) \cdot  \E[Y_{\kappa}^{(n - k + \kappa - 1)} | Y_{\kappa}^{(n - k + \kappa - 1)} < v]}{\sum_{\kappa=1}^k \mu(\kappa) \cdot Pr(Y_{\kappa}^{(n - k + \kappa - 1)} < v) }. \\
         &= \sum_{\kappa=1}^k \mu(\kappa) \cdot \E[Y_{\kappa}^{(n - k + \kappa - 1)}]. 
    \end{align*}
    Observe that $\E[Y_{\kappa}^{(n - k + \kappa - 1)}]$ is decreasing in $\kappa$.\footnote{Compare the expression for $\kappa$ and $\kappa + 1$. For any realization of $n - k + \kappa - 1$ value draws, consider adding one more value draw. If the new draw exceeds $Y_{\kappa}^{(n - k + \kappa - 1)}$, then $Y_{\kappa+1}^{(n - k + (\kappa + 1) - 1)} = Y_{\kappa}^{(n - k + \kappa - 1)}$. Otherwise, $Y_{\kappa+1}^{(n - k + (\kappa + 1) - 1)} < Y_{\kappa}^{(n - k + \kappa - 1)}$. Hence the expected value is decreasing in $\kappa$.} Thus 
    \begin{align*}
        t_{AE}(v | \mu, v) & = 
        \E_{\kappa \sim \mu}[\E[Y_{\kappa}^{(n - k + \kappa - 1)}]] \\
        & < \E_{\kappa \sim \mu'}[\E[Y_{\kappa}^{(n - k + \kappa - 1)}]] \\
        & = t_{AE}(v | \mu', v),
    \end{align*}
    where the inequality follows from strict FOSD dominance.
\end{proof}

\begin{proof}[Proof of \cref{lem:transfer-determines-entry-time}]
    Intuitively, if we pause the game at any time, we know the allocation of all remaining items is still assortatively efficient from that time onward. An agent's expected transfer is then pinned down by the Payoff Equivalence theorem. This expected transfer determines the time at which the agent must attempt to join the queue.
    
    Concretely, the given equilibrium of $Q(p)$ is assortatively efficient, so at any point in the game, the allocation of all remaining items must also be assortatively efficient. By the revelation principle, the equilibrium from this point onward thus corresponds to a direct mechanism with assortatively efficient allocation. By \cref{RET:supply_demand_uncertain}, in this direct mechanism, an agent's expected transfer is identical to her expected transfer in the continuation of the queueing game. Hence by \cref{RET:supply_demand_uncertain}, when type $v$ is in the game with beliefs $\mu$ and highest type $\hat{v}$, her expected transfer (conditional on winning) is $t_{AE}(v | \mu, \hat{v})$.

    \textbf{Forward Implication}: Suppose type $v$ attempts to join the queue at some $\tilde{t} \in \{t^+, t^-, t^{--}, \ldots\}$ By assortative efficiency, $v$ must be the highest type still remaining in the game (otherwise, all types higher than $v$ remaining in the game would also attempt to join at time $t$, yielding assortative inefficiency). Thus type $v$'s expected payment conditional on winning is $t_{AE}(v | \mu_{\tilde{t}}, v)$. But if he attempts to enter at time $\tilde{t}$, his payment conditional on winning is $t$, hence $t = t_{AE}(v | \mu_{\tilde{t}}, v)$.\footnote{Note that the agent cannot be playing a strategy that mixes between joining at time $\tilde{t}$ and joining at a later time. If she did, she must join before all other agents with probability $1$ (by assortative efficiency), hence she would always prefer to join at the later time instead of at $\tilde{t}$.}

    \textbf{Reverse Implication}: Suppose $t_{AE}(v | \mu_{\tilde{t}}, v) = t$ for some $\tilde{t} \in \{t^+, t^-, t^{--}, ...\}$. Then type $v$ must be the highest type remaining in the game at time $\tilde{t}$: otherwise, the highest type remaining in the game $\hat{v}$ would have expected payment (conditional on winning) of $t_{AE}(\hat{v} | \mu_{\tilde{t}}, \hat{v}) > t$, which is a contradiction because type $\hat{v}$ can no longer pay more than $t$ in the queueing game. Therefore, by \cref{RET:supply_demand_uncertain}, type $v$ has expected payment (conditional on winning) in the queueing game of $t_{AE}(v | \mu_{\tilde{t}}, v) = t$. If type $v$ waits until any time $t - \delta < t$ before attempting to join the queue, his expected payment conditional on winning will be strictly lower than $t$, contradicting the prior sentence. We conclude that type $v$ must attempt to join the queue at time $t$.
\end{proof}

\begin{proof}[Proof of \cref{lem:hazard-rate-efficiency}]
    Let $\vartheta$ denote the inverse hazard rate of the value distribution. By \cite{haro08} Lemma 2.6, the expected total surplus from allocation rule $x$ is $\sum_{i=1}^n\E_{\bm{v}}[\vartheta(v_i) x_i(\bm{v})] = \E_{\bm{v}}[\sum_{i=1}^n\vartheta(v_i) x_i(\bm{v})]$. 
    
    Suppose that the hazard rate of the value distribution is decreasing, so $\vartheta$ is increasing. For any realization of values $\bm{v}$, an assortatively efficient mechanism allocates goods to the agents with the highest values of $v_i$ and hence the highest values of $\vartheta(v_i)$. Thus an assortatively efficient mechanism yields weakly higher surplus than any other mechanism. Furthermore, if $\vartheta$ is strictly increasing, an assortatively efficient mechanism yields higher surplus than any mechanism that is not assortatively efficient.

    Suppose instead that the hazard rate of the value distribution is increasing, so $\vartheta$ is decreasing. For any realization of values $\bm{v}$, an assortatively efficient mechanism allocates goods to the agents with the highest values of $v_i$. These agents now have the lowest values of $\vartheta(v_i)$. Hence an assortatively efficient mechanism yields the \emph{lowest} expected utility of any mechanism that always allocates the good. Hence any mechanism with a different allocation rule which always allocates the good will provide a weakly higher surplus. This inequality is strict when $\vartheta$ is strictly decreasing.
\end{proof}

\begin{proof}[Proof of \cref{prop:entry-cost-no-information}]
    Because there are no messages, the game is equivalent to a discriminatory price auction in which each agent may choose to pay cost $c$ to enter the auction and submit a sealed bid (his entry time). By Topkis's theorem, there exists some threshold type $v^R$ such that agents enter if and only if their value is above $v^R$. By continuity, a type-$v^R$ agent must be indifferent between entering or not entering, and therefore obtains zero expected utility if he enters. If $v^R \leq c$, then he obtains a strictly negative expected utility: he always pays $c$ to enter, and with some probability less than $1$ he receives an item with value at most $c$. Hence $v^R > c$.
\end{proof}

\begin{proof}[Proof of \cref{prop:entry-cost-ae}]
    One equilibrium is the following. Agents with values less than $c$ do not participate. All other agents write down a bid $\beta(v)$ given by their bidding function in a sealed-bid discriminatory price auction with reserve price $c$. In the auction equilibrium, the bidding function is given by
    \begin{align*}
        \beta(v) = \E[\max(c, Y_k^{(n-1)}) | \max(c, Y_k^{(n-1)}) < v].
    \end{align*}
    Participating agents join the queue at time $\beta(v) - c$, if the announcer has not yet announced that the queue is full. Otherwise, they do not attempt to join the queue.

    Agents' expected allocations and payoffs are equivalent to their expected allocations and payoffs in the equilibrium of the sealed-bid discriminatory price auction with reserve price of $c$. There are also no profitable deviations that do not mimic another type. Hence agents have no profitable deviations. Furthermore, the equilibrium bidding functions are monotonically increasing for agents with values above $c$, hence the equilibrium is assortatively efficient.

    We next show that this is the unique symmetric equilibrium. In any symmetric equilibrium of the game, multiple types cannot pool by bidding the same value, because any such type could profitably deviate to increase their bid by some $\varepsilon > 0$. Furthermore, agents cannot play mixed strategies in equilibrium. Suppose some agent mixed between bids $b$ and $b' > b$. By Topkis's theorem, no other type can bid in the interval $[b, b']$. But then the agent should strictly prefer bidding $b$ to bidding $b'$. We conclude that agents' strategies are pure and fully separating. By Topkis's theorem, agents' bidding strategies are therefore strictly increasing in their values. Because the equilibrium is symmetric, the allocation is assortatively efficient. By the Payoff Equivalence theorem, any equilibrium bidding strategy must therefore correspond to $\beta$.
\end{proof}

\begin{proof}[Proof of \cref{prop:entry-cost-sudden-bad-news}]
    Fix information policy $p$ such that $p$ always announces when the queue is full ($p$ may make other announcements as well). 
    We will consider two games and show that any assortatively efficient equilibrium in the first game has a corresponding equilibrium in the second game.
    The first game $(\mathcal{Q}_c(p), F)$ is the entry cost queueing game in which agents have values drawn from atomless $F$ with full support on $[0, \bar{v}]$ and must pay cost $c>0$ to attempt to enter the queue. The second game $(\mathcal{Q}(p), G)$ is the standard queueing game in which agents do not pay an entry cost and have values drawn from distribution $G$ defined by $G(v) = F(v+c)$.\footnote{Note that this allows negative values in the support, but we otherwise define the standard queueing game identically.} Intuitively, the second game eliminates the entry cost but reduces each player's value by $c$. 
    
    Assume for the sake of contradiction that there exists some assortatively efficient equilibrium of $(\mathcal{Q}_c(p), F)$ in which sudden bad news is released with positive probability; let this equilibrium strategy profile be denoted $S^c = \times S_i^c$. Each agent $i$ follows strategy $S_i^c: (h^M, t, v) \mapsto b$, with types $v\leq c$ choosing never to join the queue. Let $S$ be a strategy profile of $(\mathcal{Q}(p),G)$ defined by an agent with value $v$ using the strategy of an agent with value $v + c$ in the strategy profile $S^c$ of $(\mathcal{Q}_c(p),F)$. That is, for each player $i$,
    \begin{align*}
        S_i(h^M,t,v) = S_i^c(h^M,t,v+c). 
    \end{align*}

    We claim that if all agents play according to strategy profile $S$ in $(\mathcal{Q}(p),G)$, then (i) the allocation is assortatively efficient, (ii) sudden bad news is released with positive probability, and (iii) $S$ is an equilibrium. This will contradict a slightly modified version of \cref{thm:FOSD-assortative-efficiency}.
    
    Given value profile $\bm{v}=\{v_1,v_2,...,v_n\}$ for all players' types, let $\bm{v}+c$ denote the value profile $\{v_1+c,v_2+c,...,v_n+c\}$. The type realizations $\bm{v}$ in $(\mathcal{Q}(p),G)$ map bijectively to the type realizations $\bm{v}+c$ in $(\mathcal{Q}_c(p),F)$, with the same probability of occurrence.\footnote{More formally, the pushforward of the measure defined by $G$ by the function $\bm{v} \mapsto \bm{v}+c$ is exactly the measure defined by $F$.} Because type $v_i$ in $(\mathcal{Q}(p),G)$ follows the same strategy as type $v_i+c$ in $(\mathcal{Q}_c(p),F)$, the realized strategies also map bijectively. Hence the path of entry times is identical between the conjectured equilibrium $S$ and the equilibrium $S^c$. 
    Thus when agents play according to $S$, (i) the allocation is assortatively efficient and (ii) $p$ announces sudden bad news with positive probability (because $p$ is identical in both games).

    It remains to show that $S$ constitutes an equilibrium. We claim that if $v_i$ has a profitable deviation from the conjectured equilibrium $S$ in $(\mathcal{Q}(p),G)$, then type $v_i+c$ has a profitable deviation from the known equilibrium $S^c$ in $(\mathcal{Q}_c(p),F)$.
    
    Consider a strategy $\tilde{S}_i$ adopted by type $v_i$ in the game $(\mathcal{Q}(p),G)$ when all other agents follow strategy $S_{-i}$. 
    Without loss of generality, we may restrict deviations to those in which the agent does not join the queue after the announcer declares that the queue is full.\footnote{Joining after the queue is full does not change an agent's payoff in $(\mathcal{Q}(p),G)$, but in $(\mathcal{Q}_c(p),F)$, the agent would lose utility.}
    
    Now suppose type $v_i+c$ plays according to strategy $\tilde{S}_i$ in $(\mathcal{Q}_c(p),F)$. Then, abusing notation slightly, the sequence of entry times and announcements given by $(\tilde{S}_i(v_i),S_{-i}(\bm{v}_{-i}))$ in the game $(\mathcal{Q}(p),G)$ is identical to the sequence of entry times and announcements given by $(\tilde{S}_i(v_i+c),S^c_{-i}(\bm{v}_{-i}+c))$ in the game $(\mathcal{Q}_c(p),F)$, for any type profile $\bm{v}_{-i}$. Hence type $v_i+c$ wins the good in $(\mathcal{Q}_c(p),F)$ in the exact same cases in which type $v_i$ wins the good in $(\mathcal{Q}(p),G)$, at the exact same price (plus the entry-cost $c$), and never pays an entry cost when she does not obtain the good.\footnote{Because the equilibrium $S^c$ of $(\mathcal{Q}_c(p),F)$ is assortatively efficient, at most a measure-zero mass of agents attempts to join at any time; hence if agent $i$ attempts to join the queue under this proposed deviation, he is successful with probability $1$.} 
    Because $\bm{v}\mapsto \bm{v}+c$ is a bijection that preserves probabilities, the agents' expected payoffs are 
    identical in both games.\footnote{The same logic applies whether $\tilde{S}_i$ is a proposed deviation or the conjectured equilibrium action $S_i$.} Hence if $\tilde{S}_i$ is a profitable deviation in $(\mathcal{Q}(p),G)$ when agents play $S_{-i}$, it is a profitable deviation in $(\mathcal{Q}_c(p),F)$ when agents play $S^c_{-i}$, a contradiction. Hence we have shown (iii), that $S$ is an equilibrium of $(\mathcal{Q}(p),G)$.
    
    We conclude that $S$ is an assortatively efficient equilibrium of $(\mathcal{Q}(p),G)$ in which sudden bad news is released with positive probability. Now observe that if $S$ is an equilibrium of $(\mathcal{Q}(p),G)$ it is also an equilibrium of $(\mathcal{Q}(p),H)$, where $H$ is the lower truncation of $G$ at $0$ (that is, any agent with negative value is assigned value $0$). Then applying the following lemma (\cref{lem:zero-atom-assortative-inefficiency}) delivers the result.

        \begin{lemma}\label{lem:zero-atom-assortative-inefficiency}
        Fix information policy $p$ and let $F$ be a distribution with full support on $[0,\bar{v}]$ and no atoms at positive values. Then any equilibrium of $\mathcal{Q}(p)$ with a positive probability of revealing sudden bad news is assortatively inefficient. 
    \end{lemma}

    The lemma is identical to \cref{thm:FOSD-assortative-efficiency} except that it allows for $F$ with an atom at $0$. 

    \begin{proof}
        The proof of \cref{thm:FOSD-assortative-efficiency} only uses the fact that $F$ is atomless twice: when proving \cref{RET:supply_demand_uncertain} (the Payoff Equivalence theorem), and 
        \cref{lem:t_AE-continuous} (to show the continuity of $t_{AE}$). It therefore suffices to show that neither result breaks when $F$ has an atom at zero. For the Payoff Equivalence theorem, an agent of type $0$ will get a payoff that is less than or equal to $0$, and thus the best they can do is choose to never enter. The rest of the requirements of Payoff Equivalence are unaffected by the atom at $0$. For \cref{lem:t_AE-continuous}, we only use the fact that $F$ is atomless in the final line, when showing that the final equation is continuous. An atom at $0$ does not affect the argument. Further, $t_{AE}(v|\mu, v)$ is still strictly increasing in $v$ since an agent of type $v = 0$ still chooses never to enter and thus receives payoff $0$. The rest of the proof of \cref{thm:FOSD-assortative-efficiency} goes through unchanged. 
    \end{proof}
\end{proof}

\begin{proof}[Proof of \cref{lem:hazard-rate-efficiency-all-mechs}]
    Because the mechanism is individually rational and makes no positive transfers, all types $v<0$ have utility zero and are never allocated the good. Let $x(\theta)$ denote the ex-interim probability that type $\theta$ is allocated the good in equilibrium. By the envelope theorem,
    \begin{align*}
        u(v) = \int_{\underline{v}}^{v} x(\theta) d\theta = \int_0^{v}x(\theta)d\theta.
    \end{align*}
    Again by standard envelope theorem arguments, the expected total surplus from the mechanism is
    \begin{align*}
        \E[U(x)] & = \int_{\underline{v}}^{\overline{v}} u(\theta) f(\theta)d\theta \\
        & = \int_{0}^{\overline{v}} u(\theta) f(\theta)d\theta \\
        & = u(\theta)F(\theta)\big|_0^{\overline{v}} - \int_0^{\overline{v}} u'(\theta) F(\theta) d\theta \\
        & = u(\overline{v}) - \int_0^{\overline{v}} x(\theta)F(\theta) d\theta \\
        & = \int_0^{\overline{v}} x(\theta) d\theta - \int_0^{\overline{v}} x(\theta)F(\theta)d\theta \\
        & = \int_0^{\overline{v}} \frac{1-F(\theta)}{f(\theta)} x(\theta) f(\theta) d\theta.
    \end{align*}
    The final expression is maximized pointwise by allocating the good to the agents with the highest inverse hazard rate, $\frac{1-F(\theta)}{f(\theta)}$. If the inverse hazard rate is weakly (strictly) increasing, then the assortatively efficient allocation rule has weakly (strictly) higher expected total surplus than any other allocation rule.
\end{proof}

\section{Continuous Bad News}

Our proof strategy for \cref{prop:cont-bad-news} is as follows: first we present the bidding functions and show that they result in an assortatively efficient allocation. As a result, we can immediately rule out any behavior that always imitates the strategy of a given type; if such a deviation were profitable, it would be profitable in any assortatively efficient incentive compatible mechanism.\footnote{Formally, the result follows from the revelation principle.} By the same logic, after a message arrives, the best an agent can do is follow the strategy prescribed by their type (or join immediately if they have a value higher than that of the highest possible type still in the game). This is because after a message arrives, it is common knowledge that agents have value distributed independently and identically on $U([0, \hat{v}])$, so one equilibrium is to enter according to \eqref{discrimbid}—precisely our prescribed bidding function.
That only leaves deviations before a message arrives, but after $\tau$. We analytically verify the remaining deviation case. 

\begin{proof}[Proof of \cref{prop:cont-bad-news}]
  There are three bidding functions that characterize the equilibrium. The initial bidding function $b_I$ is the bidding function before time $\tau$. The ``no news'' bidding function $b_{\text{NN}}$ describes how agents behave after time $\tau$, conditional on having not received a message so far. The ``yes news'' bidding function $b_{\text{YN}}$ describes how agents behave after they receive a message. 
  In other words: before $\tau$ agents enter according to $b_I$, after $\tau$ but before any message is received agents enter according to $b_{\text{NN}}$, and if a message is ever received, agents enter according to $b_{\text{YN}}$.

  Let $\hat{v} = 1/2$. Note that $\hat{v}$ is the last type to enter before $\tau$: $b_I(1/2) = 2/9 = \tau$. Further, define
  \begin{align*}
      \gamma(v) := \sqrt{144 v^4-216 v^3+33 v^2+72 v+16}.
  \end{align*}
  Then the equilibrium bidding functions are:
  \begin{align}
    b_{\text{I}}(v) & = \frac{v(3-2v)}{3(2-v)} \\
    b_{\text{YN}}(v) &= \frac{v(3\hat{v} - 2v)}{3(2\hat{v} - v)} \\
    b_{\text{NN}}(v) &= \frac{\gamma(v) - 12v^2 + 9v - 4}{36 (1-v)}.
  \end{align}

  We call the strategy profile above $s$, and the collective strategy profile $S = \bigtimes s$. 
  First, we show that $S$ results in an assortatively efficient allocation, and that $S$ is feasible—agents are never prescribed to enter earlier than the current time. 
  To see that $S$ is feasible, observe that 
  \begin{enumerate}[(i)]
      \item  $b_I(\hat{v}) = b_{\text{NN}}(\hat{v})= \tau > b_{\text{YN}}(\hat{v})$
      \item $b_I$ is strictly increasing
      \item $b_{\text{NN}}$ and $ b_{\text{YN}}$ are strictly increasing on $[0, \hat{v}]$ with $b_{\text{NN}}(v) > b_{\text{YN}}(v)$. 
  \end{enumerate} 
  By (i), any agent with type higher than $\hat{v}$ enters the game before time $\tau$. By (ii), after time $\tau$, no type is called to enter at time $t > \tau$. By (iii), after a message is received, agents' prescribed bids decrease. Also note that because each function is strictly increasing in the relevant domain, each has a well-defined inverse function. 
  
  Suppose a message arrives at time $\tau'$. Agents with values greater than $b_{\text{NN}}^{-1}(\tau')$ have already entered the queue; all remaining agents have values of at most $b_{\text{NN}}^{-1}(\tau')$. Following $s$, agents bid $b_{\text{YN}}(v) < b_{\text{NN}}(v) < \tau'$, so bids are feasible. Further, $S$ yields an assortatively efficient allocation because each bidding function is strictly increasing.
  
  To rule out deviations, we take the perspective of an agent with type $v$ and consider different forms of deviations. 
  We begin by demonstrating that he cannot profitably deviate by following the strategy of type $v'\neq v$.
  To do this, we show that if each type follows strategy $s$, his expected allocation and transfer is identical to his expected allocation and transfer in an incentive compatible (IC) and individually rational (IR) assortatively efficient mechanism $M_{AE}$. Hence if he could profitably play $s(v')$, he would also have a profitable deviation in an IC mechanism, a contradiction.

  First observe that expected allocations are the same as in $M_{AE}$ because the claimed equilibrium is assortatively efficient. Next observe that $t_{AE}(v)$, an agent's expected transfers conditional on winning in $M_{AE}$, is simply our expression for $b_{\text{I}}(v)$.\footnote{By \cref{discrimbid}.} Hence if $v \geq \hat{v}$, their expected transfer and allocation under strategy profile $S$ is identical to that of $M_{AE}$. If $v \leq \hat{v}$ we compute his expected transfer as follows, in two cases.

  If no message arrives before time $b_{\text{NN}}(v)$, the agent will enter at time $b_{\text{NN}}(v)$. The probability that no message arrives before $b_{\text{NN}}(v)$ and that he wins an item is:
  
  \begin{align*}
      \mathbb{P}(\text{Win, NN}; v) = 2 \hat{v} (1 - \hat{v}) \left( \frac{v}{\hat{v}} \right) \cdot 1 + \hat{v}^2 \left( \frac{b_{\text{NN}}(v)}{\tau} \right) \left( 1 - \left(1 - \frac{v}{\hat{v}} \right)^2 \right).
  \end{align*}

  The first term represents the case in which one other bidder arrived before time $\tau$. The second term represents the case in which no other bidders arrived before time $\tau$, but the message $m$ did not arrive before time $b_{\text{NN}}(v)$.

  If instead the message arrives before $b_{\text{NN}}(v)$, the agent will attempt to enter the queue at time $b_{\text{YN}}(v) < b_{\text{NN}}(v)$. The probability that the message arrives before $b_{\text{NN}}(v)$ and the agent wins an item is:

  $$
  \mathbb{P}(\text{Win, YN};v) = \hat{v}^2 \left( 1 - \frac{b_{\text{NN}}(v)}{\tau} \right) \left( 1 - \left(1 - \frac{v}{\hat{v}} \right)^2 \right).
  $$

    In this expression, we condition on the probability that no other bidder arrived before time $\tau$, then on the message arriving before time $b_{\text{NN}}(v)$, and finally on the bidder having one of the two highest values.

  The agent's expected equilibrium transfer is:
  \begin{align*}
    \mathbb{P}(\text{Win, NN}; v) \cdot b_{\text{NN}}(v) + \mathbb{P}(\text{Win, YN}; v) \cdot b_{\text{YN}}(v) 
    &= \frac{v(3 - 2v)}{3(2 - v)} \cdot \left(1 - (1 - v)^2\right) \\
    & = t_{AE}(v) \cdot \left(1 - (1 - v)^2\right).
  \end{align*}
  
  Recall that $1-(1-v)^2$ is the probability a type-$v$ agent wins an item in $M_{AE}$, so the final expression represents the expected transfer of a type-$v$ agent in $M_{AE}$. We conclude that for both $v<\hat{v}$ and $v>\hat{v}$, an agent's expected allocation and transfer under $S$ is the same as in $M_{AE}$.

  Finally, we show that agents do not have a profitable deviation that consists of strategies not used by any type. As usual, we may restrict our attention to pure-strategy deviations. 
  We can immediately rule out some deviations. 
  
  First consider deviations in which agents bid higher than $\tau$ (at some time different from $b_I(v)$). Such a deviation would be equivalent to imitating some type $v' > \hat{v}$, which we have already shown cannot be profitable.\footnote{This strategy is equivalent to imitating $v'$ regardless of what the deviation would prescribe after time $\tau$: by the time the clock reaches $\tau$, the agent would have already entered. Also note that entering before the highest type agent would enter is strictly worse.} 

  Next consider an agent's behavior after a message has arrived at time $\tau' < \tau$. We claim that after the message has arrived, a type $v$ agent optimally bids according to
  \begin{align} \label{modified_bYN_for_deviators}
      b(v) &= 
      \begin{cases}
          b_{\text{YN}}(v) &\text{ if } v\leq b_{\text{NN}}^{-1}(\tau') \\
          b_{\text{YN}}(b_{\text{NN}}^{-1}(\tau')) &\text{ if } v > b_{\text{NN}}^{-1}(\tau'). 
      \end{cases}
  \end{align}
  In other words, an agent should follow his prescribed strategy if his type ``should'' still be in the game (i.e., he isn't a high type deviating downward), and otherwise he should enter when the highest type still in the game enters.

  No agent should bid more than $b_{\text{YN}}(b_{\text{NN}}^{-1}(\tau'))$: All types still in the game have values $v < b_{\text{NN}}^{-1}(\tau')$, so bidding more increases payments without changing the probability of winning. 
  We next rule out deviations to bids within $[0, b_{\text{YN}}(b_{\text{NN}}^{-1}(\tau'))]$. 
  After receiving a message, it is common knowledge that no agent arrived before $\tau$, which means agents' values are distributed uniformly on $[0, \hat{v}]$. Note that $b_{\text{YN}}$ is the expected payment conditional on winning in $M_{AE}$ when all agents have values distributed uniformly on $U([0,\hat{v}])$. Note further that all agents with types greater than $b_{\text{YN}}(b_{\text{NN}}^{-1}(\tau'))$ have already entered. 
  
  Suppose the agent deviates to enter at time $t \in [0, b_{\text{YN}}(b_{\text{NN}}^{-1}(\tau'))]$. Then $t = b_{\text{YN}}^{-1}(v')$ for some $v'$, and hence the agent obtains the same expected payoff and transfer as agent $v'$ would in $M_{AE}$ in which all agents have values distributed uniformly on $U([0,\hat{v}])$. If $v < b_{\text{NN}}^{-1}(\tau')$, then the type-$v$ agent could make the same deviation in $M_{AE}$, so it cannot be profitable. 
  If $v > b_{\text{NN}}^{-1}(\tau')$, then by Topkis's theorem, the agent's optimal action is at least the optimal action of $b_{\text{NN}}^{-1}(\tau')$—that is, $b_{\text{YN}}(b_{\text{NN}}^{-1}(\tau'))$. By the logic above, he will never bid above this value, so bidding exactly $b_{\text{YN}}(b_{\text{NN}}^{-1}(\tau'))$ is optimal. Hence once a message arrives, following \eqref{modified_bYN_for_deviators} is optimal. 

  Finally, we consider deviations in which an agent enters after $\tau$ but before news arrives (by the previous argument, after a message arrives, the best they can do is to follow \eqref{modified_bYN_for_deviators}). Such a deviation does not imitate the strategy profile of another type, and thus we must rule it out explicitly. 
  
  We consider two cases: the agent deviates to bid above $b_{\text{NN}}(v)$, and the agent deviates to bid below $b_{\text{NN}}(v)$.

  In the former case, the agent bids $b_{\text{NN}}(v')$ for some $v' > v$. If news arrives before $b_{\text{NN}}(v')$, then as shown above, his optimal strategy is to bid $b_{\text{YN}}(v)$.\footnote{Note that for this case, $b_{\text{YN}}(v) < b_{\text{YN}}(b_{\text{NN}}^{-1}(\tau'))$ if $\tau'$ arrives before $b_{\text{NN}}(v')$.}

  We then compute his expected payoff with the following expression, letting $\pi^U(v';v)$ denote the  expected utility of an agent with type $v$ who bids according to $b_{\text{NN}}(v')$:
  \begin{align}\label{util_NN_upward_deviation}
  \pi^U(v';v) & =  \hat{v}^2 \cdot \left(1 - \frac{b_{\text{NN}}(v')}{\tau}\right) \cdot \left(1 - \left( \frac{\hat{v} - v}{\hat{v}} \right)^2 \right) \cdot (v - b_{\text{YN}}(v)) \\
  &\quad + \hat{v}^2 \cdot \left( \frac{b_{\text{NN}}(v')}{\tau} \right) \cdot \left(1 - \left( \frac{\hat{v} - v'}{\hat{v}} \right)^2 \right) \cdot (v - b_{\text{NN}}(v')) \nonumber \\
  &\quad + 2\hat{v}(1 - \hat{v}) \cdot \left( \frac{v'}{\hat{v}} \right) \cdot (v - b_{\text{NN}}(v')). \nonumber
  \end{align}

  We have broken the expression into three terms. 
  The first term captures the case in which neither of the other two agents arrived in the queue before $\tau$, and a message arrived before the agent entered (i.e., before time $b_{\text{NN}}(v')$). In this case, the agent then joins at the time prescribed by $b_{\text{YN}}$. 
  The second term represents the case in which neither other agent arrived in the queue before $\tau$, but no message arrives before $b_{\text{NN}}(v')$. The third term represents the case in which at least one agent joined the queue before time $\tau$ (note that if both other agents joined the queue before $\tau$, then our agent gets payoff $0$).

  It remains to show that for any $v$ and $v'>v$ in the appropriate domains, (\ref{util_NN_upward_deviation}) is strictly less than the utility from following $s$. We subtract the latter expression from the former and show that the expression is strictly negative by finding the roots of the equation. 
  An agent's additional payoff from deviating to $v' > v$ is:
  \begin{align}\label{upward_util_minus_AE_util}
    \pi^U(v';v)-\pi^U(v;v) = \frac{(v-v')^2}{48 (v'-1)}\times \big( &-2 v \left(-12 v'^2+\gamma(v')+9 v'-4\right) +48 v'^3 \\
    &-72 v'^2+\left(19-4 \gamma(v')\right) v'+3 \left(\gamma(v')+4\right)\big). \nonumber
  \end{align}

    Note that \eqref{upward_util_minus_AE_util} is continuous in the relevant domain, and only has positive roots at $v'=v$ and $v'=\frac{1}{24v}(-6 v^2+\sqrt{36 v^4-108 v^3+249 v^2+36 v+4}+9 v+2)$. The second root is always greater than $\hat{v}=1/2$, so we can evaluate \eqref{upward_util_minus_AE_util} when $v' = \hat{v}$ to determine its sign.
    \begin{align*}
        \pi^U(\hat{v};v)-\pi^U(v;v) = -\frac{1}{12} (1-2 v)^2 (2-v) < 0.
    \end{align*}
    
    We conclude that the additional deviation payoff is negative for $v' \in (v,\hat{v}]$, so no upward deviation is profitable. 
  
    The last case is to consider when the agent deviates in the no-news case to bid below $b_{\text{NN}}(v)$. If a message arrives before they enter, then without loss we assume they bid according to \eqref{modified_bYN_for_deviators}, as we earlier showed that following \eqref{modified_bYN_for_deviators} is optimal.

    We let $\pi^D(v';v)$ denote the payoff from deviating downward to bid $b_{\text{NN}}(v')$. Then
    \begin{align*}
    \pi^D(v’;v) =\ &
    \hat{v}^2 \cdot \left(1 - \frac{b_{\text{NN}}(v)}{\tau}\right) \cdot \left(1 - \left(\frac{\hat{v} - v}{\hat{v}}\right)^2\right) \cdot (v - b_{\text{YN}}(v)) \\
    &+ \hat{v}^2 \cdot \int_{b_{\text{NN}}(v’)}^{b_{\text{NN}}(v)}
    \frac{1}{\tau} \cdot
    \left(1 - \left( \frac{\hat{v} - b_{\text{NN}}^{-1}(\tau’)}{\hat{v}} \right)^2 \right) \cdot
    \left(v - b_{\text{YN}}(b_{\text{NN}}^{-1}(\tau’)) \right)
    \, d\tau’ \\
    &+ \hat{v}^2 \cdot \left(\frac{b_{\text{NN}}(v’)}{\tau}\right) \cdot \left(1 - \left(\frac{\hat{v} - v’}{\hat{v}}\right)^2\right) \cdot (v - b_{\text{NN}}(v’)) \\
    &+ 2\hat{v}(1 - \hat{v}) \cdot \left(\frac{v’}{\hat{v}}\right) \cdot (v - b_{\text{NN}}(v’)).
    \end{align*}    

    The first term represents the case when the message arrives before $b_{\text{NN}}(v)$. The second term represents the case when the message arrives after $b_{\text{NN}}(v)$ but before $b_{\text{NN}}(v')$: in this case, the agent follows \eqref{modified_bYN_for_deviators} and 
    bids $b_{\text{YN}}(b_{\text{NN}}^{-1}(\tau’))$ after the message—the same bid that the highest possible remaining type in the game would make—if the message arrives at time $\tau'$. The third term represents the case when no other bidders arrived before $\tau$, but the message did not arrive before $b_{\text{NN}}(v')$. The fourth term represents the case when exactly one bidder arrived before $\tau$.

    We show the derivative of $\pi^D$ with respect to $v'$ is positive when $v' < v$:
    \begin{align}\label{util_NN_downward_deviation}
    \pi^D_1(v';v) = \frac{\left(24 v'^3-30 v'^2+\left(9-2 \gamma(v')\right) v'+\gamma(v')+4\right) (v-v')}{8 (1-v')}.
    \end{align}

    Equation  \eqref{util_NN_downward_deviation} is continuous over the relevant domain. Its roots are $v'=v$ and $v' = \frac{1}{12} \left(3 \pm \sqrt{57}\right)$, so the sign of \eqref{util_NN_downward_deviation} is the same for all $v'\in [0, v]$. When $v'=0$, $\pi^D_1(v' = 0;v) = v > 0$. Hence $\pi^D_1(v';v)>0$ for all $v'<v$ and we conclude that the deviation is not profitable.
  
  Thus there are no profitable deviations and $S$ constitutes an assortatively efficient equilibrium, as claimed.

\end{proof}

\end{appendix}

\bibliographystyle{chicago} 
\bibliography{queuebib.bib}

\begin{thebibliography}{}

\bibitem[\protect\citeauthoryear{Anunrojwong, Iyer, and Manshadi}{Anunrojwong
  et~al.}{2023}]{aniyma23}
Anunrojwong, J., K.~Iyer, and V.~Manshadi (2023).
\newblock Information design for congested social services: Optimal need-based
  persuasion.
\newblock {\em Management Science\/}~{\em 69\/}(7), 3778--3796.

\bibitem[\protect\citeauthoryear{Baccara, Lee, and Yariv}{Baccara
  et~al.}{2020}]{baleya20}
Baccara, M., S.~Lee, and L.~Yariv (2020, 07).
\newblock Optimal dynamic matching.
\newblock {\em Theoretical Economics\/}~{\em 15\/}(3), 1221--1278.

\bibitem[\protect\citeauthoryear{Bulow and Klemperer}{Bulow and
  Klemperer}{1994}]{bukl94}
Bulow, J. and P.~Klemperer (1994, 02).
\newblock Rational frenzies and crashes.
\newblock {\em Journal of Political Economy\/}~{\em 102\/}(1).

\bibitem[\protect\citeauthoryear{Che and Tercieux}{Che and
  Tercieux}{2024}]{chte24}
Che, Y.-K. and O.~Tercieux (2024, 01).
\newblock Optimal queue design.
\newblock {\em Journal of Political Economy, Forthcoming\/}.

\bibitem[\protect\citeauthoryear{Condorelli}{Condorelli}{2012}]{cond12}
Condorelli, D. (2012, 07).
\newblock What money can't buy: Efficient mechanism design with costly signals.
\newblock {\em Games and Economic Behavior\/}~{\em 75\/}(2), 613--624.

\bibitem[\protect\citeauthoryear{Dworczak}{Dworczak}{2023}]{dwor23}
Dworczak, P. (2023, 08).
\newblock Equity-efficiency trade-off in quasi-linear environments.
\newblock {\em Working paper\/}.

\bibitem[\protect\citeauthoryear{Gretschko, Rasch, and Wambach}{Gretschko
  et~al.}{2014}]{grrawa14}
Gretschko, V., A.~Rasch, and A.~Wambach (2014, 01).
\newblock On the strictly descending multi-unit auction.
\newblock {\em Journal of Mathematical Economics\/}~{\em 50}, 79--85.

\bibitem[\protect\citeauthoryear{Hartline and Roughgarden}{Hartline and
  Roughgarden}{2008}]{haro08}
Hartline, J.~D. and T.~Roughgarden (2008, 05).
\newblock Optimal mechanism design and money burning.
\newblock {\em Proceedings of the Fortieth Annual ACM Symposium on Theory of
  Computing\/}, 75--84.

\bibitem[\protect\citeauthoryear{Hassin and Haviv}{Hassin and
  Haviv}{2003}]{haha03}
Hassin, R. and M.~Haviv (2003).
\newblock {\em To Queue or Not to Queue: Equilibrium Behavior in Queueing
  Systems}.
\newblock Boston, MA: Springer US.

\bibitem[\protect\citeauthoryear{Holt and Sherman}{Holt and
  Sherman}{1982}]{hosh82}
Holt, C.~A. and R.~Sherman (1982, 04).
\newblock Waiting-line auctions.
\newblock {\em Journal of Political Economy\/}~{\em 90\/}(2), 280--294.

\bibitem[\protect\citeauthoryear{Krishna}{Krishna}{2009}]{kris09}
Krishna, V. (2009, 08).
\newblock {\em Auction Theory\/} (2 ed.).
\newblock 30 Corporate Drive, Suite 400, Burlington, MA 01803, USA: Academic
  Press.

\bibitem[\protect\citeauthoryear{Leshno}{Leshno}{2022}]{lesh22}
Leshno, J. (2022, 12).
\newblock Dynamic matching in overloaded waiting lists.
\newblock {\em American Economic Review\/}~{\em 112\/}(12), 3876--3910.

\bibitem[\protect\citeauthoryear{Leswing}{Leswing}{2017}]{news_lesw17}
Leswing, K. (2017).
\newblock Apple says to line up early if you want an iphone x.
\newblock {\em Business Insider\/}.

\bibitem[\protect\citeauthoryear{Murra-Anton and Thakral}{Murra-Anton and
  Thakral}{2024}]{muth24}
Murra-Anton, Z. and N.~Thakral (2024).
\newblock The public housing allocation problem.
\newblock {\em Working paper\/}.

\bibitem[\protect\citeauthoryear{Naor}{Naor}{1969}]{naor69}
Naor, P. (1969, 01).
\newblock The regulation of queue size by levying tolls.
\newblock {\em Econometrica\/}~{\em 37\/}(1), 15--24.

\bibitem[\protect\citeauthoryear{Price}{Price}{2017}]{new_pric17}
Price, R. (2017).
\newblock People have been queueing outside the apple store for days to buy the
  iphone x.
\newblock {\em Business Insider\/}.

\bibitem[\protect\citeauthoryear{Reny}{Reny}{1999}]{reny99}
Reny, P.~J. (1999, 09).
\newblock On the existence of pure and mixed strategy nash equilibria in
  continuous games.
\newblock {\em Econometrica\/}~{\em 67\/}(5), 1029--1056.

\bibitem[\protect\citeauthoryear{Taylor, Tsui, and Zhu}{Taylor
  et~al.}{2003}]{tats03}
Taylor, G.~A., K.~K. Tsui, and L.~Zhu (2003).
\newblock Lottery or waiting-line auction?
\newblock {\em Journal of Public Economics\/}~{\em 87\/}(5), 1313--1334.

\end{thebibliography}

\end{document}